\documentclass[runningheads]{llncs}
\usepackage{bbding}
\usepackage{etoolbox}
\pdfoutput=1
\usepackage[T1]{fontenc}
\def\doi#1{\href{https://doi.org/\detokenize{#1}}{\url{https://doi.org/\detokenize{#1}}}}
\usepackage{graphicx}
\advance\textwidth4mm 
\advance\hoffset-2mm
\advance\textheight4mm
\advance\voffset-2mm

\usepackage{amsmath,amssymb,mathtools}
\usepackage{algorithm}
\usepackage{algorithmicx}
\usepackage[noend]{algpseudocode}
\usepackage{booktabs}
\usepackage{tikz}
\usepackage{nicefrac}
\usepackage{xparse}
\usepackage{xfrac}
\usepackage{multirow}
\usepackage{tikz}\usetikzlibrary{automata, positioning, arrows}
\tikzset {
	->,
	>=stealth',
	node distance=3cm,
	every state/.style={thick, fill=gray!10},
	initial text=$ $,
}

\renewcommand{\circ}{\bigcirc}

\newtoggle{arxiv}

\togglefalse{arxiv}

\newcommand{\Naturals}{\mathbb{N}}
\newcommand{\Reals}{\mathbb{R}}

\newcommand{\abs}[1]{\lvert #1 \rvert}
\DeclareMathOperator*{\argmax}{arg\, max}
\DeclareMathOperator*{\argmin}{arg\, min}

\newcommand{\eqdef}{\vcentcolon=}

\newcommand{\Dist}{\mathsf{Dist}}

\newcommand{\post}{\mathsf{Post}}
\newcommand{\diff}{\mathsf{diff}}
\newcommand{\lfp}{\mathsf{lfp}}

\newcommand{\qee}{\hfill$\triangle$}

\renewcommand{\Box}{\square}
\renewcommand{\circ}{\bigcirc}
\newcommand{\GG}{\mathcal{G}}
\newcommand{\SG}{\GG}
\newcommand{\Gsigma}{\GG^{(\sigma,\cdot)}}
\newcommand{\Gtau}{\GG^{(\cdot,\tau)}}
\newcommand{\Gst}{\GG^{(\sigma,\tau)}}
\newcommand{\randomRandomOutcome}{\GG_\mathsf{Algo}}
\newcommand{\connectedSG}{\mathsf{\GG_{reach}}}

\newcommand{\states}{\mathsf{S}}
\newcommand{\minStates}{\states_{\circ}}
\newcommand{\maxStates}{\states_{\Box}}
\newcommand{\initstate}{s_{0}}

\newcommand{\state}{s}
\newcommand{\action}{a}
\newcommand{\act}{\mathsf{A}}
\newcommand{\actions}{\act}
\newcommand{\Av}{\mathsf{Av}}
\newcommand{\trans}{\delta}
\newcommand{\stateMac}{s}

\newcommand{\sink}{\mathsf{z}}
\newcommand{\sinks}{\mathsf{Z}}
\newcommand{\target}{\mathsf{f}}
\newcommand{\targets}{\mathsf{F}}

\newcommand{\Bop}{\mathcal{B}}
\newcommand{\BopD}{\mathcal{B}^{\mathsf{D}}}
\newcommand{\unknown}{\states^?}

\newcommand{\Paths}{\mathsf{Paths}}

\newcommand{\probability}{\mathcal{P}}
\newcommand{\val}{\mathsf{V}}
\newcommand{\lb}{\mathsf{L}}
\newcommand{\ub}{\mathsf{U}}

\newcommand{\SI}{\textsc{SI}}
\newcommand{\LPSI}{\SI_\textsc{LP}}
\newcommand{\SISI}{\SI_\textsc{SI}}
\newcommand{\TLPSI}{\textsc{T}\LPSI}
\newcommand{\TSISI}{\textsc{T}\SISI}

\newcommand{\VI}{\textsc{VI}}
\newcommand{\OVI}{\textsc{OVI}}

\newcommand{\BVI}{\textsc{BVI}}

\newcommand{\WP}{\textsc{WP}}
\newcommand{\PTVI}{\textsc{PTVI}}
\newcommand{\PTBVI}{\textsc{PTBVI}}

\newcommand{\realcs}{REAL}

\newcommand{\many}{RANDOM}

\newcommand{\bestExit}{\mathsf{bexit}}
\newcommand{\drawcirc}{\node[draw,circle,minimum size=.7cm, outer sep=1pt]}
\newcommand{\drawbox}{\node[draw,rectangle,minimum size=.7cm, outer sep=1pt]}
\newcommand{\drawdummy}{\node[minimum size=0,inner sep=0]}

\usepackage{algorithm}
\usepackage{nccmath}
\usepackage[noend]{algpseudocode}
\usepackage{comment}
\usepackage{amssymb}
\usepackage{tabularx}
\usepackage{pgffor}
\usepackage{tikz}
\usepackage{tikz}\usetikzlibrary{automata, positioning, arrows, graphs, quotes, calc,angles, shapes.multipart, trees, patterns,decorations.pathreplacing,decorations.markings}
\usepackage{placeins}
\usepackage{subfig}
\tikzset{node distance=2.5cm, 
   every state/.style={minimum size=0pt, fill=gray!10,rectangle, 
   align=center}, 
   initial text={}.
   every picture/.style={>=stealth'},
   brace/.style={decorate,decoration=brace}, semithick}

\newcommand{\para}[1]{\smallskip\noindent{\bf#1}}
\usepackage{listings}
\let\llncssubparagraph\subparagraph
\let\subparagraph\paragraph
\usepackage[compact]{titlesec}
\usepackage{titlesec}
\let\subparagraph\llncssubparagraph

\titlespacing*{\section}{0pt}{2.5ex plus 1ex minus 1ex}{2ex plus 0.5ex minus 0.5ex}
\titlespacing*{\subsection}{0pt}{2.25ex plus 0.5ex minus 1ex}{1.5ex plus 0.5ex minus 0.5ex}
\titlespacing*{\subsubsection}{0pt}{1ex plus 0.5ex minus 0ex}{1ex}

\makeatletter
\RequirePackage[bookmarks,unicode,colorlinks=true]{hyperref}%
   \def\@citecolor{blue}%
   \def\@urlcolor{blue}%
   \def\@linkcolor{blue}%

\def\orcidID#1{\smash{\href{http://orcid.org/#1}{\protect\raisebox{-1.25pt}{\protect\includegraphics{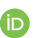}}}}}
\makeatother
\begin{document}
\title{Optimistic and Topological Value Iteration for Simple Stochastic Games
}
\titlerunning{Optimistic and Topological VI for SSGs}

\author{Muqsit Azeem\orcidID{0000-0003-4532-8344} \and
Alexandros Evangelidis$^{\text{\href{mailto:alexandros.evangelidis@tum.de}{\Envelope}}}$\,\orcidID{0000-0003-4032-3042} \and
Jan K{\v{r}}et\'insk\'y\orcidID{0000-0002-8122-2881}
\and Alexander Slivinskiy\orcidID{0000-0002-0856-106X}\and Maximilian Weininger\orcidID{0000-0002-0163-2152}}
\authorrunning{Azeem, Evangelidis, K{\v{r}}et\'insk\'y, Slivinskiy, Weininger}
%
%
\institute{Technical University of Munich, Munich, Germany\\
\email{firstname.lastname@tum.de}}
\maketitle              

\vspace{-1.75em}
\begin{abstract}
   While value iteration (VI) is a standard solution approach to simple stochastic games (SSGs), it suffered from the lack of a stopping criterion. Recently, several solutions have appeared, among them also ``optimistic'' VI (OVI).
	However, OVI is applicable only to one-player SSGs with no end components. We lift these two assumptions, making it available to \emph{general SSGs}.
	Further, we utilize the idea in the context of topological VI, where we provide an efficient \emph{precise} solution. In order to compare the new algorithms with the state of the art, we use not only the standard benchmarks, but we also design a \emph{random generator} of SSGs,
   which can be biased towards various types of models, aiding in understanding the advantages of different algorithms on SSGs.

\end{abstract}

\vspace{-1.75em}

\section{Introduction}
\label{sec:intro}

\vspace{-0.5em}
\para{Stochastic games} (SGs) are a standard model for decision making in the presence of adversary and uncertainty, by combining two (opposing) non-determinisms with stochastic dynamics.
Thus, they extend both Markov decision processes (MDPs), the standard model for sequential decision making and probabilistic verification, and 2-player graph games, the standard model for reactive synthesis.
\emph{Simple stochastic games} (SSGs) \cite{condonComplexity} form an important special case where the goal is to reach a given state.
In technical terms, an SSG is a zero-sum two-player turn-based game played on a graph by Maximizer and Minimizer, who choose actions in their respective vertices (also called states).
Each action is associated with a probability distribution determining the next state to move to.
The objective of Maximizer is to maximize the probability of reaching a given target state; the objective of Minimizer is the opposite.
The interest in SSGs stems from two sources.
Firstly, solving an SSG is polynomial-time equivalent to solving perfect information Shapley, Everett and Gillette games \cite{DBLP:conf/isaac/AnderssonM09} and further important problems can be reduced to SSGs, for instance parity games, mean-payoff games, discounted-payoff games and their stochastic extensions~\cite{DBLP:journals/corr/abs-1106-1232}; yet, the complexity of solving SSGs remains a long-standing open question, known to be in $\mathbf{UP}\cap\mathbf{coUP}$~\cite{HK66}, but with polynomial-time algorithm staying elusive.
Secondly, the problem is practically relevant in verification and synthesis in stochastic environments, with many applications, e.g., \cite{DBLP:journals/algorithmica/LaValle00,DBLP:journals/fmsd/ChenFKPS13,DBLP:conf/qest/ChenKSW13,DBLP:conf/icse/CamaraMG14}, surveyed in detail in \cite{DBLP:journals/ejcon/SvorenovaK16}.
Consequently, heuristics improving performance of the algorithms for solving SSGs are also practically relevant.

\para{Algorithms} used to (approximately) solve SSGs can be divided into several classes, most notably quadratic programming (QP) and dynamic programming, the latter comprising strategy iteration (SI) and value iteration (VI).
For their practical comparison, see the recent \cite{gandalf20}.

On the one hand, only when exact solutions are required, SI is mostly used.
It provides a sequence of improving strategies and, accompanied by evaluation of Markov chains via systems of linear equations, can yield the precise result.
On the other hand, approximate solutions (with a certain imprecision) are  faster to compute and often sufficient.
For this reason, VI is the technique used in practice the most, e.g., in PRISM-games \cite{prismgames3}, although not necessarily always the best.
It gradually approximates (from below) the optimal probability to reach the target from each state.
Interestingly, until very recently no means were known to determine the current precision, and so standard implementations terminating whenever no significant improvements occur can be arbitrarily wrong \cite{bvi}.
More surprisingly, this was even the case for MDPs, i.e.,~SSGs with a single player.

In 2014 \cite{atva,bvi}, the first stopping criterion for MDPs was given, quantifying precision of the current approximation by providing also a sequence converging to the optimal probabilities from above.
The difficulty to obtain such converging upper bound arises from cyclic dependencies of the optimal probabilities in so-called end components (ECs).
For instance, an action surely self-looping on a state trivializes the equations, stating only that the probability in this state is simply equal to itself, yielding an infinity of solutions, not just the optimal one.
This issue has been solved for MDPs \cite{atva,bvi} by ``collapsing'' these ECs into single states with no loops, which corresponds to identifying cyclically dependent variables into a single one.

In 2018 \cite{KKKW18}, the idea was finally extended to SSGs, giving rise to \emph{bounded value iteration} (BVI) with the first stopping criterion for SSGs.
Note that the MDP solution could not be directly used since the analog of ECs in SSGs is more complex:
different states in an EC in an SSG can have different optimal probabilities and thus cannot be merged.
Instead, \emph{``deflating''} manipulates the values in smaller and dynamically changing ``simple ECs''.

Since the first VI stopping criterion was given for MDPs, several alternatives have been proposed, most notably \emph{sound VI} (SVI) \cite{DBLP:conf/cav/QuatmannK18} and \emph{optimistic VI} (OVI) \cite{ovi}.
However, the termination proofs of both require the MDP to contain no ECs.
They achieve this by collapsing ECs, which is not applicable to SSGs.

\para{Our contribution.} 
In this paper, we extend the idea of OVI in two ways so that we obtain algorithms for SSGs.

\paragraph{First algorithm} The idea of OVI~\cite{ovi} is to run VI (converging from below) until changes are small, then to guess slightly larger values and check whether they form an upper bound. 
If not, the process continues.
To overcome the requirement that there is no EC, we complement the procedure of \cite{ovi} with the deflating of \cite{KKKW18}.
However, to ensure monotonicity of the Bellman operator, the so-called ``simple'' ECs must be computed differently from~\cite{KKKW18}. 
While the rest of the proof is analogous to \cite{ovi}, we try to make it simpler and more elegant by separating the core idea from the practical improvements.
As a result, we obtain an OVI algorithm for SSGs.

\paragraph{Second algorithm} We consider the classic \emph{``topological''} optimization of VI~\cite{TVI1}, where the system is analysed per strongly connected component (SCC) in the bottom-up order.
While such decomposition often leads to savings in runtime and memory, also when expected accumulated rewards are considered \cite{ensure}, the imprecisions from lower SCCs propagate to the upper ones, yielding the method useless whenever the system is too deep (with as few SCCs in a row as 20) even for Markov chains, see Example~\ref{ex:tvi}.
We fix this issue by precise \emph{and} fast computations in each SCC as follows.
First, we quickly obtain an approximate solution by VI, then we optimistically guess the solution, but in contrast to OVI which guesses values, we guess optimal strategies, which turns out to require orders of magnitude fewer guesses.
If the guess is not correct, a step of SI can be cheaply performed.
This version of OVI can thus be also seen as a possible warm start~for~SI.



\paragraph{Comparison and model generation}
We compare the resulting approaches to BVI and a more recent SSG solution called ``widest path'' \cite{widest} (WP).
While there is no clear winner, we provide insights as to which algorithm to use in different settings.
As noticed already in \cite{gandalf20}, the performance of SSG algorithms is extremely sensitive to the structure of the models.
Unfortunately, there are too few realistic case studies and thus a very limited number of model structures.
Consequently, in order to be able to experimentally compare our algorithms in a reasonable way, we propose an approach for random SSG generation.
While we prove that our approach can generate every SSG, it skews towards certain types of models.
Hence we provide means for the user to skew towards model structures that they are interested in, e.g.,\ increasing or decreasing the number of SCCs.
This helps to find out which algorithms are sensitive to which model parameters, e.g.,\ amount of SCCs.
While this is only the first step towards filling this gap of random SSG (and MDP) generation, we hope to encourage more research on the topic through this~effort.


\para{Our contribution} can be summarized as follows:
\begin{itemize}
	\item We design an extension of OVI to SSGs. As a side effect, we extend OVI on MDPs, lifting the requirement of no ECs (Section~\ref{sec:ovi}).
	\item We extend the landscape by providing an efficient VI-based approach for precise solutions, using the OVI idea on strategies, rather than values (Section~\ref{sec:pt}).
	\item We provide and evaluate a random generator of SSGs, which can be biased towards various types of models (Section~\ref{sec:randomGen}).
	\item We compare the resulting methods to the state of the art (BVI, WP, SI) experimentally (Section~\ref{sec:exp}).
\end{itemize}

\para{Related work.}
Closest to our work, in the case of SSGs, is the work of \cite{KKKW18} where the first stopping criterion for VI was given. 
It extends both normal BVI~\cite{bvi} and its learning-based counterpart \cite{atva} from MDPs by incorporating the so-called deflating procedure as part of their computation.
Recently, another BVI variant for SSGs was proposed which introduces a global propagation of upper bounds~\cite{widest}.
Also, the simpler case of an SSG with one-player ECs is discussed in \cite{MatPhD}.

In general, the tools which are available for solving SGs are limited. PRISM-games \cite{prismgames3} implements the
standard VI algorithms, and it also considers other objectives apart from reachability, such as mean-payoff and ratio reward.
Further, GAVS+ \cite{DBLP:conf/tacas/ChengKLB11} is an algorithmic game solver with support for solving SSGs,
and GIST \cite{gist} allows for the qualitative verification of SGs.
%
%
%

\newpage
\section{Preliminaries}
\label{sec:prelims}

\subsection{Simple stochastic games}
\label{sec:prelims:sgs}
A probability distribution on a finite set $X$ is a mapping $\delta : X \rightarrow [0,1]$, such that $\sum_{x \in X} \delta(x)=1$. The set of all probability distributions on $X$ is denoted by $\Dist(X)$.   

\begin{definition}[Stochastic game (SG), e.g.,~\cite{DBLP:conf/dimacs/Condon90}]
    A stochastic turn-based two-player game is defined by a tuple $\mathcal{G}= \langle \states, \maxStates, \minStates, \initstate, \act, \Av, \delta \rangle$
	where $\states$ is a finite set of states partitioned into a set of \emph{Minimizer} $(\minStates)$ and \emph{Maximizer} $(\maxStates)$ states, respectively.
$s_0 \in \states$ is the initial state. $\act$ is a finite set of actions. $\Av : \states \rightarrow 2^\act$ assigns to every state a set of available actions.
Finally, $\delta : \states \times \act \rightarrow \Dist(\states)$ is the transition function.

Note that a Markov decision process (MDP) is a special case of an SG where either $\minStates=\emptyset$ or $\maxStates=\emptyset$ and a Markov chain is a special case of an MDP where in each state
there is only one available action.
\end{definition}

Without loss of generality, we assume SGs to be non-blocking, i.e., for all $s \in \states: \Av(s) \neq \emptyset$.
For convenience, we use the following notation: 
Given a state $s \in \states$ and an action $a \in \Av(s)$, the set of successor states is denoted as $\post(s,a) := \{s' \ | \ \delta(s,a,s') >0 \} $.
For a set of states $T \subseteq \states$, we use $T_\Box = T \cap \maxStates$ to denote all Maximizer states in $T$, and dually for Minimizer.
Figure~\ref{fig:sg-example} shows an example SG.

\begin{figure}[t]
\begin{center}
	\begin{tikzpicture}
		\drawdummy (init) at (0,0) {};
		\drawcirc (p) at (1,0) {$\initstate$};
		\drawbox (q) at (3,0) {$s_1$};
		\drawdummy (mid) at (4.25,0) {};
		\drawbox (1) at (5.5,0.5) {$\target$ };
		\drawcirc (0) at (5.5,-0.5)  {$\sink$};

		\draw[->] (init) to (p);
		\draw[->]  (p) to[bend left] node [midway,anchor=south] {$\mathsf{a}$}(q) ;
		\draw[->]  (q) to [bend left] node [midway,anchor=north] {$\mathsf{b}$} (p);
		\draw[-] (q) to node [midway,anchor=south] {$\mathsf{c}$} (mid) ;
		\draw[->] (mid) to node [below] {$\sfrac12$} (0);
		\draw[->] (mid) to node [above] {$\sfrac12$} (1);
		\draw[->]  (0) to[loop right]  node [midway,anchor=west] {$\mathsf{e}$} (0);
		\draw[->]  (1) to [loop right] node [midway,anchor=west] {$\mathsf{d}$} (1) ;
	\end{tikzpicture}
\end{center}
\caption{An example of an SG with $\states = \{\mathsf{\initstate},\mathsf{s_1},\target,\sink\}, \maxStates= \{\mathsf{s_1},\target\}$,
$\minStates = \{\mathsf{\initstate},\sink\}$, the initial state $\mathsf{s_0}$ and set of
actions $\act= \{\mathsf{a},\mathsf{b},\mathsf{c},\mathsf{d},\mathsf{e}\}$; $\mathsf{Av}(\mathsf{\initstate})=\{\mathsf{a}\}$
with $\trans(\mathsf{\initstate},\mathsf{a})(\mathsf{s_1})=1$;  $\mathsf{Av}(\mathsf{s_1})=\{\mathsf{b},\mathsf{c}\}$
with $\trans(\mathsf{s_1},\mathsf{b})(\mathsf{\initstate})=1$ and 
$\trans(\mathsf{s_1},\mathsf{c})(\target) = \trans(\mathsf{s_1},\mathsf{c})(\sink) = \frac{1}{2}$.
For actions with only one successor, we do not depict the transition probability~$1$.}\label{fig:sg-example}
\end{figure}

\para{Semantics: paths, strategies and the value.} 
Formally, an \emph{infinite path} $\rho$ is defined as $\rho=s_0a_0s_1a_1...\in (\states \times \act)^\omega$,
such that for every $i \in \Naturals$, $a_i \in \Av(s_i)$ and $s_{i+1} \in \post(s_i,a_i)$.
The set of all paths in an SG $\GG$ is denoted as $\Paths_\GG$.
A finite path is a prefix of an infinite path ending in a state $s$.
 
A Maximizer \emph{strategy} is a function $\sigma: \maxStates \to \act$ such that $\sigma(s) \in \Av(s)$ for all~$s$;
Minimizer strategies $\tau$ are defined analogously.  
We restrict attention to \emph{memoryless deterministic} strategies, because they are sufficient for the objective we consider~\cite{condonComplexity}.
By fixing both players' choices according to a pair of strategies $(\sigma, \tau)$, we turn an SG $\GG$ into a Markov chain $\Gst$ with state space $\states$ and the transition function $\trans^{\sigma,\tau}(s,s')= \trans(s,\sigma(s),s')$ for Maximizer states $s$ and dually for Minimizer with $\sigma$ replaced by $\tau$. 
Given a state $s$, the Markov chain $\Gst$ induces a unique probability distribution $\probability_{s}^{\sigma,\tau}$ over the set of all infinite paths~\cite[Sec. 10.1]{BK08}.

Since we consider SSGs, we complement an SG with a set of goal states $\targets \subseteq \states$ and formalize the objective of reaching $\targets$, as follows:
we denote as $\lozenge{\targets} := \lbrace \rho \ | \ \rho = s_0a_0s_1a_1... \in \Paths_\GG \land  \exists
i \in \Naturals. s_i \in \targets \rbrace$ the (measurable) set of all paths which eventually reach $\targets$.
We are interested in the \emph{value} of every state $s$, i.e.,\ the probability that $s$ reaches a goal state if both players play optimally. 
Formally, for each $s \in \states$, its value is defined as
\begin{ceqn}
	\begin{align}
		\val(s):= \sup\limits_{\sigma}\inf\limits_{\tau} \probability_{s}^{\sigma,\tau} (\lozenge \targets) = \inf\limits_{\tau}\sup\limits_{\sigma}\probability_{s}^{\sigma,\tau} (\lozenge \targets),
	\end{align}
\end{ceqn}
where the equality follows from \cite{condonComplexity}.
We use $\val: \states \to \Reals$ to denote the function that maps every $s \in \states$ to its value. 
When comparing functions $f_1, f_2: \states \to \Reals$, we use point-wise comparison, i.e.,\ $f_1 \leq f_2$ if and only if for all $s \in \states: f_1(s) \leq f_2(s)$.

\subsection{Value iteration and bounded value iteration}
To compute the value function $\val$ for an SSG, the following partitioning of the state space is useful: firstly the goal states $\targets$, secondly the set of \emph{sink states} that do not have a path to the target $\sinks = \lbrace s \in \states \mid \nexists \rho = s_0a_0s_1a_1... \in \Paths_\GG: s_0 = s \land \rho \in \lozenge \targets \rbrace$, and finally the remaining states $\states^?$. 
For $\targets$ and $\sinks$ (which can be easily identified by graph-search algorithms), the value is trivially 1 respectively 0. Thus, the computation only has to focus on $\states^?$.

The well-known approach of value iteration leverages the fact that $\val$ is the least fixpoint of the \emph{Bellman equations}, cf.~\cite{visurvey}:
\begin{ceqn}
	\begin{align}
		\val(s) = \begin{cases}
			1 & \text{if} \ s \in \targets \\
			0 & \text{if} \ s \in \sinks \\
			\max_{a \in \Av(s)} \Big(\sum_{s' \in \states} \trans(s,a,s') \cdot \val(s')\Big) & \text{if} \ s \in \maxStates^? \\
			\min_{a \in \Av(s)} \Big(\sum_{s' \in \states} \trans(s,a,s') \cdot \val(s')\Big) & \text{if} \ s \in \minStates^?
		\end{cases}   
	\end{align}
\label{eq:bellman}
\end{ceqn}

Now we define\footnote{In the definition of $\Bop$, we omit the technical detail that for goal states $s \in \targets$, the value has to remain 1. Equivalently, one can assume that all goal states are absorbing, i.e.,\ only have self looping actions.} the Bellman operator $\Bop: (\states \to \Reals) \to (\states \to \Reals)$:
\begin{ceqn}
	\begin{align}
		\mathsf{\Bop(f)}(s) = \begin{cases}
			\max_{a \in \Av(s)} \Big(\sum_{s' \in \states} \trans(s,a,s') \cdot f(s')\Big) & \text{if} \ s \in \maxStates \\
			\min_{a \in \Av(s)} \Big(\sum_{s' \in \states} \trans(s,a,s') \cdot f(s')\Big) & \text{if} \ s \in \minStates
		\end{cases}   
	\end{align}
\end{ceqn}

Value iteration starts with the under-approximation 
\[\lb_0(s) = \begin{cases} 
	1 & \text{if} \ s \in \targets\\
	0 & \text{otherwise} \end{cases}\]
and repeatedly applies the Bellman operator. Since the value is the least fixpoint of the Bellman equations and $\lb_0 \leq \val$ is lower than the value, this converges to the value in the limit~\cite{visurvey} (formally $\lim_{i \to \infty} \Bop^i(\lb_0) = \val$).

While this approach is often fast in practice, it has the drawback that it is not possible to know the current difference between $\Bop^i(\lb_0)$ and $\val$ for any given~$i$. 
To address this, one can employ \emph{bounded value iteration} (BVI, also known as interval iteration~\cite{atva,bvi,KKKW18})
It additionally starts from an over-approximation $\ub_0$, with $\ub_0(s) = 1$ for all $s \in \states$. However, applying the Bellman operator to this upper estimate might not converge to the value, but to some greater fixpoint instead, see~\cite[Section 3]{KKKW18} for an example.
The core of the problem are so called \emph{end components}.
\begin{definition}[End component (EC)]
	A set of states $T$ with $\emptyset \neq T \subseteq \states$ is an \emph{end component} if and only if there exists a set of actions $\emptyset \neq B \subseteq \bigcup_{s \in T}Av(s)$ such that:\\
	1. for each $s \in T$, $a \in B \cap \Av(s)$ we have $\post(s,a) \subseteq T$. \\
	2. for each $s, s' \in T$ there exists a finite path $\mathsf{w = sa_0...a_ns' } \in (T \times B)^* \times T$.
	
	An end component $T$ is a \emph{maximal end component} (MEC) if there is no other EC $T'$ such that $T \subseteq T'$.
\end{definition}

Intuitively, ECs can be problematic, because the over-approximation $\ub$ is higher in the EC than the value. Thus, Maximizer prefers staying in the EC and keeping the illusion of achieving the high $\ub$; it is an illusion, because staying will never reach a target, and Maximizer actually has to use some exit of the EC.
The solution proposed in~\cite{KKKW18} explicitly identifies these situations and forces all states in the EC to decrease their $\ub$ by making it depend on the best exit of the EC.
This operation is called \emph{deflating}, to evoke the impression of releasing the pressure in an EC that is bloated by having too high estimates.
To define deflating more formally, we need two definitions from~\cite{KKKW18}:  

\begin{definition}[Best exit]
	Given a set of states $T \subseteq \states$ and a function $f : \states \rightarrow \Reals$, the best exit
	according to $f$ from $T$ is defined as:\\
	$\bestExit_f (T) = \displaystyle \max_{\substack{s \in T_\Box, a \in \Av(s)\\ \post(s,a) \nsubseteq T}} \ \Big( \sum_{s' \in \states} \trans(s,a,s') \cdot f(s,a) \Big) $,
	
	\noindent with the convention that $\max_\emptyset=0$.
\end{definition}

\begin{definition}[Simple end component (SEC)]
	An EC T is a simple end component (SEC) if for all $s \in T$, $\val(s)= \bestExit_\val (T)$
\end{definition}

In SSGs, states in an EC can have different values. Thus, it is necessary to find the SECs. In these simple sub-parts of the EC all states have the same value, namely that of the best exit. 
By setting the over-approximation to $\bestExit_\ub(T)$ for each SEC $T$ (additionally to applying $\Bop$), we ensure that it converges to the value~\cite{KKKW18}.
As a final complication, computing SECs is difficult, since they depend on the value $\val$ that we want to compute. The solution of~\cite{KKKW18} is to use the current under-approximation $\lb$ to guess which states form a SEC and as $\lb$ converges to $\val$ in the limit, eventually we guess correctly. 

Thus, we can augment the Bellman operator with additional deflating and define an operator $\BopD_\lb: (\states \to \Reals) \to (\states \to \Reals)$. Note that it depends on an $\lb$ to guess the SECs. Given a function $\ub$, it proceeds as follows:
\begin{itemize}
	\item Apply a Bellman update $\Bop(\ub)$.
	\item Guess the SECs according to $\lb$ by using~\cite[Algorithm 2]{KKKW18}.
	\item For each SEC $T$ and all states $s \in T$, set $\ub(s) = \min(\ub(s),\bestExit_\ub(T))$. The $\min$ is only to ensure monotonicity.
\end{itemize}

In summary, BVI computes two sequences: 
the sequence of lower bounds $\lb_i = \Bop^i(\lb)$ for $i \in \Naturals$ and an additional sequence of upper bounds $\ub_i = (\BopD_{\lb})^i(\ub)$. Note that for the $i$-th application of $\BopD_\lb$, it uses the current lower bound $\lb_i$.
Both sequences converge to the value $\val$ in the limit~\cite[Theorem 2]{KKKW18}.
This allows to terminate the algorithm when the difference between the lower and upper bound is less than a pre-defined precision $\varepsilon$ and obtain an $\varepsilon$-approximation of the value.
%

\section{Optimistic Value Iteration}
\label{sec:ovi}

The idea of optimistic value iteration (OVI, \cite{ovi}) is to leverage the fact that classic VI (only from below) typically converges quickly to the correct value.
Indeed, the following ``naive'' stopping criterion results in an approximation that is $\varepsilon$-close in all available realistic case studies:
stop when for all $s \in \states$ applying the Bellman update does not result in a big difference, i.e. $\diff(\lb(s),\Bop(\lb)(s)) < \varepsilon$,
where we use $\diff(\text{old},\text{new}) = \text{new}-\text{old}$ to denote the absolute difference between two numbers\footnote{One can also use the relative difference, i.e. $\diff(\text{old},\text{new}) = \frac{\text{new} - \text{old}}{\text{new}}$.}. 
However, the naive stopping criterion can also terminate early when the estimate still is arbitrarily wrong~\cite{bvi}.

OVI first performs classic VI with the naive stopping criterion, optimistically hoping that it will terminate close to the value.
Additionally, it uses a \emph{verification phase}, where it checks whether the result of VI  was indeed correct. 
If it was, OVI terminates with the guarantee that we are $\varepsilon$-close to the value.
Otherwise, if the result of VI cannot be verified, OVI continues VI with a higher precision $\varepsilon'$. By repeating this, at some point $\varepsilon'$ is so small that when VI terminates, OVI can verify that the result is $\varepsilon$-precise.

Our version of OVI for SSGs is given in Algorithm~\ref{alg:ovi}. 
Lines~\ref{line:ovi:whileVI}-\ref{line:ovi:updateL} are the classic VI, Lines~\ref{line:ovi:guessU}-\ref{line:ovi:updateUmonotonic} the verification phase. 
Concretely, in the verification phase we first guess a candidate upper bound $\ub$ (Line \ref{line:ovi:guessU}), so that the difference between $\lb$ and $\ub$ is small enough that, if $\ub$ indeed is an upper bound, we could terminate.
Formally, for all $s \in \states$, $\ub(s) = \diff^+ (L(s))$, 
where $\diff^+_\varepsilon(x) = \begin{cases}
	0 &\mbox{ if}~x = 0\\
	x + \varepsilon &\mbox{ otherwise}
\end{cases}$ for absolute difference\footnote{$\diff^+_\varepsilon(x) = x * (1+\varepsilon)$ for relative difference.}.
Then we apply the Bellman operator once (Line~\ref{line:ovi:updateU}) and check whether $\BopD_\lb(\ub)\leq\ub$ (Line~\ref{line:ovi:verifTermCondition}). If that holds, we know (by arguments from lattice theory) that $\val \leq \ub$, i.e.\ that $\ub$ is a valid upper bound on the value. 
Thus, since $\lb$ and $\ub$ are $\varepsilon$-close to each other and $\lb \leq \val \leq \ub$, we return an $\varepsilon$-approximation of the value (Line \ref{line:ovi:return}).
\textbf{The key difference} between the original algorithm for MDPs and the extension to SSGs is that we do not use $\Bop$ in Line~\ref{line:ovi:updateU} any more, but the Bellman operator with additional deflating $\BopD$.
On MDPs, the termination of OVI relied on the assumption that there were no ECs. This is justified, since in MDPs one can remove the ECs by ``collapsing'' them beforehand, cf.~\cite{atva,bvi}. 
On SSGs, collapsing is not possible~\cite{KKKW18}, which is why we need the new operator. 

We have addressed the case that the guessed $\ub$ can indeed be verified as an upper bound. 
In the other case where we are not (yet) able to verify it, Algorithm \ref{alg:ovi} continues applying $\BopD_\lb$ for a finite number of times (we chose $\frac{1}{\varepsilon'}$, Line \ref{line:ovi:forVerifPhase}). If for all iterations we cannot verify $\ub$ as an upper bound, the precision $\varepsilon'$ for the naive stopping criterion is increased (we chose $\frac{\varepsilon'}{2}$) and we start over (Line \ref{line:ovi:tryAgain}).

\begin{theorem}\label{thm:ovi}
	Given an SSG $\GG$ and a lower bound $\lb_0 \leq \val$, OVI$(\GG,\lb_0,\varepsilon,\varepsilon)$ terminates and returns $(\lb, \ub)$ such that $\lb \leq \val \leq \ub$ and $\diff(\ub(s),\lb(s)) \leq \varepsilon$ for all $s \in \states$.
\end{theorem}

\begin{algorithm}[t]
	\caption{Optimistic value iteration for SSGs.} \label{alg:ovi}
	\begin{algorithmic}[1]
		\Require SSG $\GG$, lower bound $\lb \leq \val$, precision $\varepsilon>0$ and naive precision $\varepsilon'>0$
		\Ensure $(\lb, \ub)$ such that $\lb \leq \val \leq \ub$ and $\diff(\ub(s),\lb(s)) \leq \varepsilon$ for all $s \in \states$
		\Procedure{OVI}{$\GG,\lb,\varepsilon,\varepsilon'$}
		\smallskip
		\Statex $\triangleright$ Classic VI with naive convergence criterion
		\While{for some state $s \in \states: \diff(\lb(s),\Bop(\lb)(s)) > \varepsilon'$}\label{line:ovi:whileVI}
		\State $\lb \gets \Bop(\lb)$\label{line:ovi:updateL}
		\EndWhile
		\smallskip
		
		\Statex $\triangleright$ Verification phase
		\State $\ub \gets \{s \mapsto \diff^+_\varepsilon(\lb(s)) \mid s \in \states\}$ \Comment{Guess candidate upper bound} \label{line:ovi:guessU}
		\For {$\frac 1 {\varepsilon'}$ times} \label{line:ovi:forVerifPhase}
		\State $\ub' \gets \BopD_\lb(\ub)$\label{line:ovi:updateU}
		\If {$\ub' \leq \ub$}\label{line:ovi:verifTermCondition}
		\State \textbf{return} $(\lb,\ub)$ \Comment{Found inductive upper bound}\label{line:ovi:return}
		\EndIf	
		\State For all $s \in \unknown: \ub(s) \gets \min(\ub(s),\ub'(s))$ \Comment{Ensure monotonicity} \label{line:ovi:updateUmonotonic}
		\EndFor
		\State \textbf{return} OVI($\GG,\lb,\varepsilon,\frac{\varepsilon'}{2}$) \Comment{Try again with more precision} \label{line:ovi:tryAgain}
		\EndProcedure
	\end{algorithmic}
\end{algorithm}

Our formulation of Algorithm~\ref{alg:ovi} is simpler than \cite[Algorithm 2]{ovi}, since we include only the key parts that are necessary for the proof of Theorem~\ref{thm:ovi} (provided in Appendix~\ref{app:ovi}).
Below we comment on three ways in which our algorithm can be changed, following the ideas of~\cite[Algorithm 2]{ovi}.
All these changes are not necessary for correctness or termination, but they can practically improve the algorithm.

\begin{enumerate}
	\item We can include a check $\BopD(\ub) \geq \ub$. It allows to detect whether $\ub \leq \val$, i.e.\ $\ub$ actually is a lower bound on the value. In that case, one can immediately terminate the verification phase and use $\ub$ as the new $\lb$. We include this improvement in our implementation, and it is used in almost every unsuccessful verification phase.
	\item The original version continues to update the lower bound during the verification phase. This is used for an additional breaking condition if the lower bound crossed the upper bound in some state. For clarity of presentation, we chose to separate concerns and only update the upper bound in the verification phase. This improvement never made a significant difference in our experiments.
	\item The original version used Gau{\ss}-Seidel VI, cf. \cite[Section 3.1]{ovi}, for both the lower and the upper bound. Our implementation allows the user to select whether to use classic or Gau{\ss}-Seidel VI.
\end{enumerate}

\FloatBarrier
\section{Precise topological value iteration}\label{sec:pt}

Topological value iteration (TVI, \cite{TVI1}) is a variant of VI that does not solve the whole game at once, but rather proceeds piece by piece. This can speed up convergence and help with memory issues.
Concretely, it uses the insight that the strongly connected components (SCCs) of an SSG always form a directed acyclic graph. 
Thus, one can first solve the bottom SCCs, i.e.\ the last in the topological ordering, and then proceed backwards one SCC by the next, relying on the results of the already computed successor SCCs.
This idea is not restricted to VI algorithms, but can also be used for other solutions methods like strategy iteration (SI) and quadratic programming~\cite{gandalf20}.

The evaluation of~\cite{gandalf20} showed that this can be quite useful in some cases, but also much slower in other, possibly even running into time outs on models where the normal algorithms succeed.
The reason for this is a complex problem that did not occur in the proof of correctness, as it is related to machine precision:
SCCs are not solved precisely, but only with $\varepsilon$-precision. That means that SCCs which are considered later in the computation have suboptimal information about their exits. 
This not only slows down convergence, but can even aggregate and lead to precision problems and non-termination when there is a chain of many SCCs, as we show in the following example.

\begin{algorithm}[t]
	\caption{Precise topological value iteration} \label{alg:pt}
	\begin{algorithmic}[1]
		\Require SSG $\GG$
		\Ensure The precise value $\val$ for all states in $\GG$
		\Procedure{PTVI}{$\GG$}
		\For{every SCC $T$ in reverse topological ordering}\label{line:pt:forSCC}
		\State Select arbitrary $\varepsilon$
		\State $\lb,\ub \gets$ computed by some VI-algorithm with precision $\varepsilon$ \label{line:pt:vi}
		\State Compute strategies $\sigma,\tau$ which are optimal according to $\lb$ and $\ub$ \label{line:pt:strats}
		\State Precisely compute the value $\val_{\Gst}$ of $T$ in the Markov chain $\Gst$ \label{line:pt:MCsolve}
		
		\If{For all $s \in T: 
			\begin{cases}
				\sigma(s) \in \argmax_{a \in \act(s)} \val_{\Gst}(s,a) &\text{if } s \in \maxStates\\
				\tau(s) \in \argmin_{a \in \act(s)} \val_{\Gst}(s,a) &\text{if } s \in \minStates
			\end{cases}$} \label{line:pt:if}
		\State \textbf{Return} $\val_{\Gst}$ as value for $T$. \label{line:pt:return}
		\Else
		\State Apply strategy iteration, using $\sigma$ or $\tau$ as initial strategy.\label{line:pt:SI}
		\EndIf
		\EndFor
		\EndProcedure
	\end{algorithmic}
\end{algorithm}

\begin{example}\label{ex:tvi}
	To exemplify TVI and show when its precision problems occur, we consider an SSG that is a chain of $n$ SCCs, each with one state. Every state either loops or continues to the next state, both with probability $0.5$. At the end of the chain, we go to the goal with $0.6$ and to the sink with $0.4$.
	
	Formally, $\states=\maxStates = \{t,z,s_0,s_1,\ldots,s_n\}$, where $s_0$ is the initial state and $t \in \targets$ is the only goal state.
	There only is one action $a$, so $\Av(s)=\act=\{a\}$ for all states $s\in\states$.
	For every $s_i$ with $i<n$, we have $\trans(s_i,a,s_i)=\trans(s_i,a,s_{i+1})=0.5$ and for $s_n$, we have $\trans(s_n,a,t)=0.6$ and $\trans(s_n,a,z)=0.4$. Both states $t$ and $z$ are absorbing, so they loop with probability one.
	
	Running topological bounded VI on this SSG, we first solve the bottom SCCs, i.e. $t$ and $z$, and  (by graph algorithms) infer their values of 1 and 0, respectively.
	Then we solve the SCC $\{s_n\}$ and set both its bounds to $0.6$.
	Next, for the SCC $\{s_{n-1}\}$ bounded VI returns an $\varepsilon$-precise result, as with the self-loop the precise value is only obtained in the limit. 
	Using precision of $\varepsilon=10^{-6}$, the resulting interval is $[0.5999994277954102,0.6000003814697266]$. 
	Now the imprecisions start to add up:
	when solving the next SCC $\{s_{n-2}\}$, we depend on the imprecise bounds for $\{s_{n-1}\}$.
	Thus, the progress we make in every Bellman update is smaller.
	This not only slows down convergence, but it also leads to the first $\varepsilon$-precise interval being $[0.5999994099140338,0.6000003933906441]$. 
	So when BVI for the SCC $\{s_{n-2}\}$ terminates, both the lower and the upper bound are less precise than in the previous SCC.
	In state $s_{n-19}$, this imprecision has aggregated such that the computation is stuck at the interval $[0.5999994000000000,0.6000004000000001]$, where the difference is larger than~$\varepsilon$. 
	Even though theoretically we make progress with a Bellman update, this progress is smaller than machine precision, so practically we can neither converge nor terminate.
	
	Note that the SSG in this example is a Markov chain, so this problem occurs not only in SSGs, but already in Markov chains and~MDPs.\qee

\end{example}

We address this problem by introducing the precise-topological-optimization (PTVI, see Algorithm~\ref{alg:pt}).
The idea of PTVI is that, after an SCC has been solved with $\varepsilon$-precision (Line~\ref{line:pt:vi}), we first extract the strategies for both players from the result (Line~\ref{line:pt:strats}) and then compute the exact value of all states in the SCC under this pair of strategies (Line~\ref{line:pt:MCsolve}). Finally, we use a simple local check to verify that this is indeed the optimal value (Line~\ref{line:pt:if}).
If it is, we return the precise values that the next SCCs can safely depend on (Line~\ref{line:pt:return}). 
If it is not, then we have to continue with some precise solution method (Line~\ref{line:pt:SI}). Since we have just extracted near-optimal strategies, it makes sense to continue with SI, see e.g.,\ \cite[Section 3.2]{gandalf20}.
For details on the selection of the strategies and the proof of Theorem~\ref{thm:pt}, see Appendix~\ref{app:pt}.
\begin{theorem}\label{thm:pt}
	Algorithm~\ref{alg:pt} returns the precise solution $\val$.
\end{theorem}

The strength of PTVI is the simple local check that allows it to conclude that the estimates for an SCC are precise. 
It relies on guessing both strategies. This differs from guessing an upper bound, as OVI does; or guessing one strategy, as in SI with a warm start~\cite[Section 4.3]{gandalf20}.
We emphasize that even using the classical naive stopping criterion in Line~\ref{line:pt:vi}, this local check succeeded on more than 99\% of the case studies, and thus the additional steps of Line~\ref{line:pt:SI} are almost never necessary. Using bounded VI in Line~\ref{line:pt:vi}, we immediately succeeded on all case studies.
In contrast, the first verification phase of OVI --- having the same estimates and thus ``information'' as PTVI has when performing the local check --- succeeded only for 15\% of the random case studies; in 85\% of the random cases as well as several larger real case studies OVI had to perform additional verification phases.

Note that PTVI can be seen from different directions.
(i) It is a practical fix of TVI~\cite{TVI1}. 
(ii) It is a new way to make classical VI return a precise result, which is more efficient than running for an exponential number of steps and rounding as described in~\cite{visurvey}. 
(iii) It is a warm start for SI, in the seldom case that the SI phase of the algorithm (Line~\ref{line:pt:SI}) is necessary. 
(iv) Just like OVI, it optimistically iterates the lower bound and then uses guessing to verify this guess. However, unlike OVI it produces a precise result, albeit at the cost of solving a Markov chain precisely; and it uses the information available at the time of guessing more efficiently, succeeding on the first check more often OVI.
\section{Random Generation of Simple Stochastic Games} \label{sec:randomGen}
In order to properly evaluate and compare our algorithms, we need a diverse set of benchmarks. However, to the best of our knowledge, there are only 12 SSG case studies modelling real world problems and 3 handcrafted models for theoretical corner cases.
Since the underlying structure of a model greatly affects the runtime of algorithms~\cite{gandalf20}, only scaling these few models is insufficient.
Thus, we propose an algorithm for random generation of SSG case studies, which enables us to test our algorithms on a broader spectrum of models.

Moreover, as we are interested in the relation between our verification algorithms and certain features of the model structure, our implementation also allows for skewing the probability distribution towards models that exhibit certain features.
This is very useful, since it allows us to test our algorithms on models whose features were not considered before (e.g., large number of actions per state, etc.).
In particular, we provide: \textbf{(i)} parameters to tune features that can be affected by parameters of single states (e.g., the size, percentage of Minimizer states, actions per state, etc.).
For example, if for each state the probability of being a Maximizer or Minimizer state is equal, we get 50\% Minimizer states on average. Similarly, by choosing a high probability of adding another action to a state during the generation, we obtain states with up to 90 actions and an average around 7;
\textbf{(ii)} more involved guidelines to affect features which depend on the interactions of several states (e.g., the number and size of SCCs and ECs, etc.).
Intuitively, to obtain an SCC or MEC of a certain size $n$, we have to restrict the choice of successors during the transition or action generation to ensure that there are $n$ strongly connected states.

We provide a detailed description of our random generation algorithm in Appendix~\ref{app:randomGen}.
There, we also prove that it can generate every possible SSG with positive probability and describe and discuss the aforementioned guidelines.
Additionally, we give a detailed analysis of model features for all random case studies used in the evaluation, as well as a comparison to the features of the real case studies in
Appendix~\ref{app:modelAnalysis}.
\section{Experiments}\label{sec:exp}

In this section we talk about the practical evaluation of our algorithms and the comparison to the state of the art. 
First, we describe the setup in Section~\ref{sec:exp:setup}. 
Then we give a general overview in Section~\ref{sec:exp:overview} before analyzing the algorithms' performance in more detail in Sections~\ref{sec:exp:precise-details} and~\ref{sec:exp:approx-details}.

\subsection{Experimental setup}\label{sec:exp:setup}

\para{Algorithms.}
Our implementation is based on PRISM-games~\cite{prismgames3} and available at \url{https://github.com/ga67vib/Algorithms-For-Stochastic-Games}.

We compare to the following algorithms from related work: classical value iteration (\VI,~\cite{visurvey}), bounded value iteration ($\BVI$,~\cite{KKKW18}) and the improvement of bounded value iterations based on widest paths ($\WP$,~\cite{widest}).
Moreover, as a representative of a competitor yielding a precise result, we implemented a precise variant of strategy iteration ($\SI$), which relies on linear programming for solving the opponent MDP.

The new algorithms are optimistic VI ($\OVI$, Section~\ref{sec:ovi}) and the precise topological version of VI (Section~\ref{sec:pt}). For the latter, we give two variants with different stopping criteria in Line~\ref{line:pt:vi} of the algorithm: $\PTVI$ uses the naive criterion and $\PTBVI$ the $\varepsilon$-guaranteed one.
Finally, we consider several optimizations, but their analysis is delegated to Appendix~\ref{app:exp:opts} due to space constraints.
Quite surprisingly, for all optimizations, their impact can be positive or negative on different models.

\para{Case studies.}
We consider case studies from three different sources: 
(i) all real case studies that were already used in~\cite{gandalf20}, and are mainly part the PRISM benchmark suite~\cite{PRISMben};
(ii) several handcrafted corner case models: haddad-monmege (the adversarial example from~\cite{bvi}), BigMec and MulMec (a single big MEC or a long chain of many small MECs from~\cite{gandalf20}), as well as two new models to analyse the behaviour of \OVI\ and one large model with many SCCs;
(iii) randomly generated models as discussed in Section~\ref{sec:randomGen}. 
Note that throughout our experiments, we omitted models solved by pre-computations.

\para{Technical details.}
We conducted the experiments on a server with 64 GB of RAM and a 3.60GHz Intel CPU running Manjaro Linux. We always use a precision of $\varepsilon=10^{-6}$. The timeout was set to 15 minutes and
the memory limit was 6 GB for all models except for large models ($\geq$ 1,000,000 states). For the large models, the timeout was set to 30 minutes and the memory limit to 36 GB.

\subsection{Overview}\label{sec:exp:overview}

\begin{figure}[h!]
	\centering
	\subfloat[
	Real models]{{
			\includegraphics[width=0.48\textwidth]{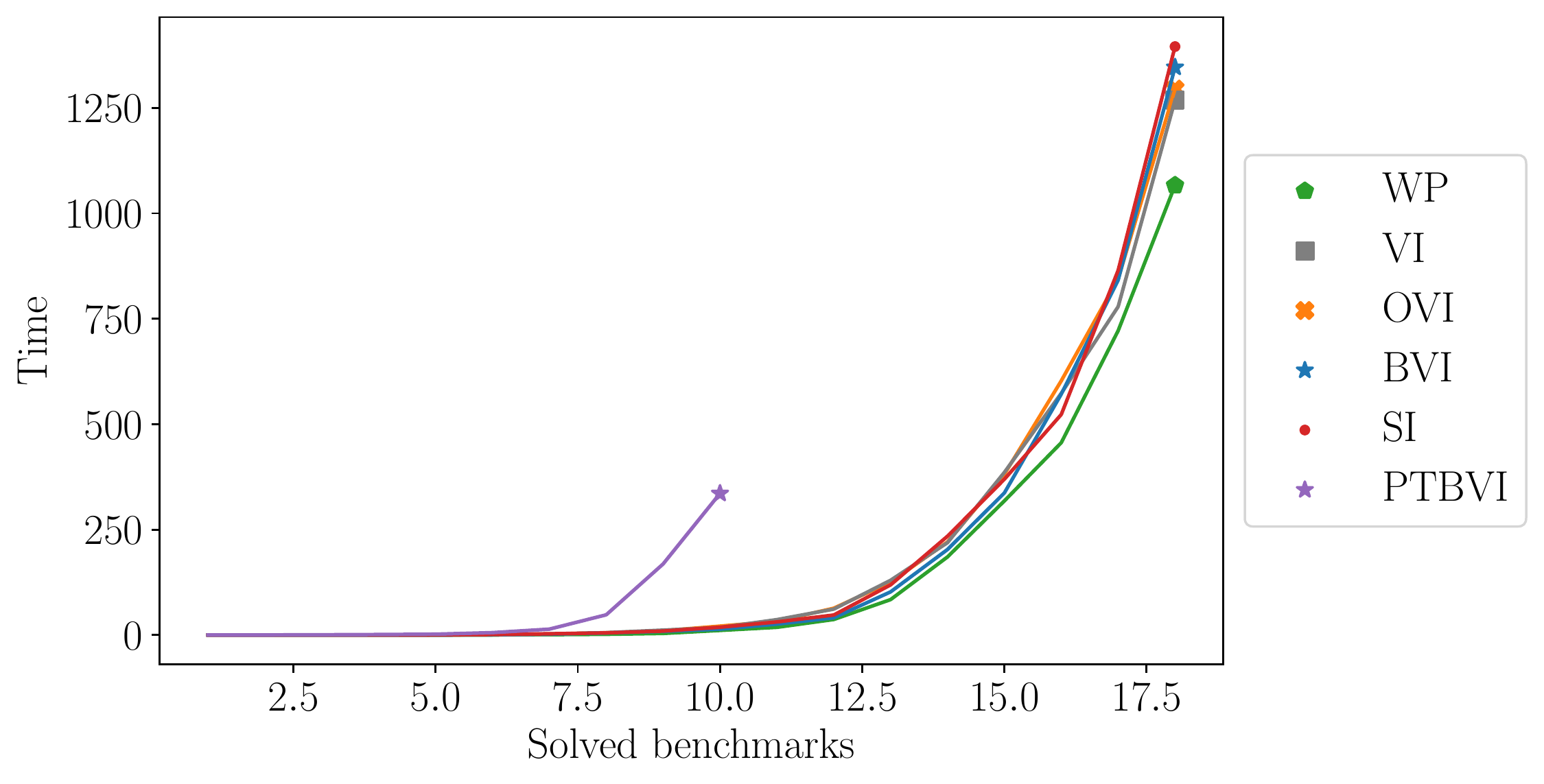} }}%
	\
	\subfloat[
	Random models]{{\includegraphics[width=0.48\textwidth]{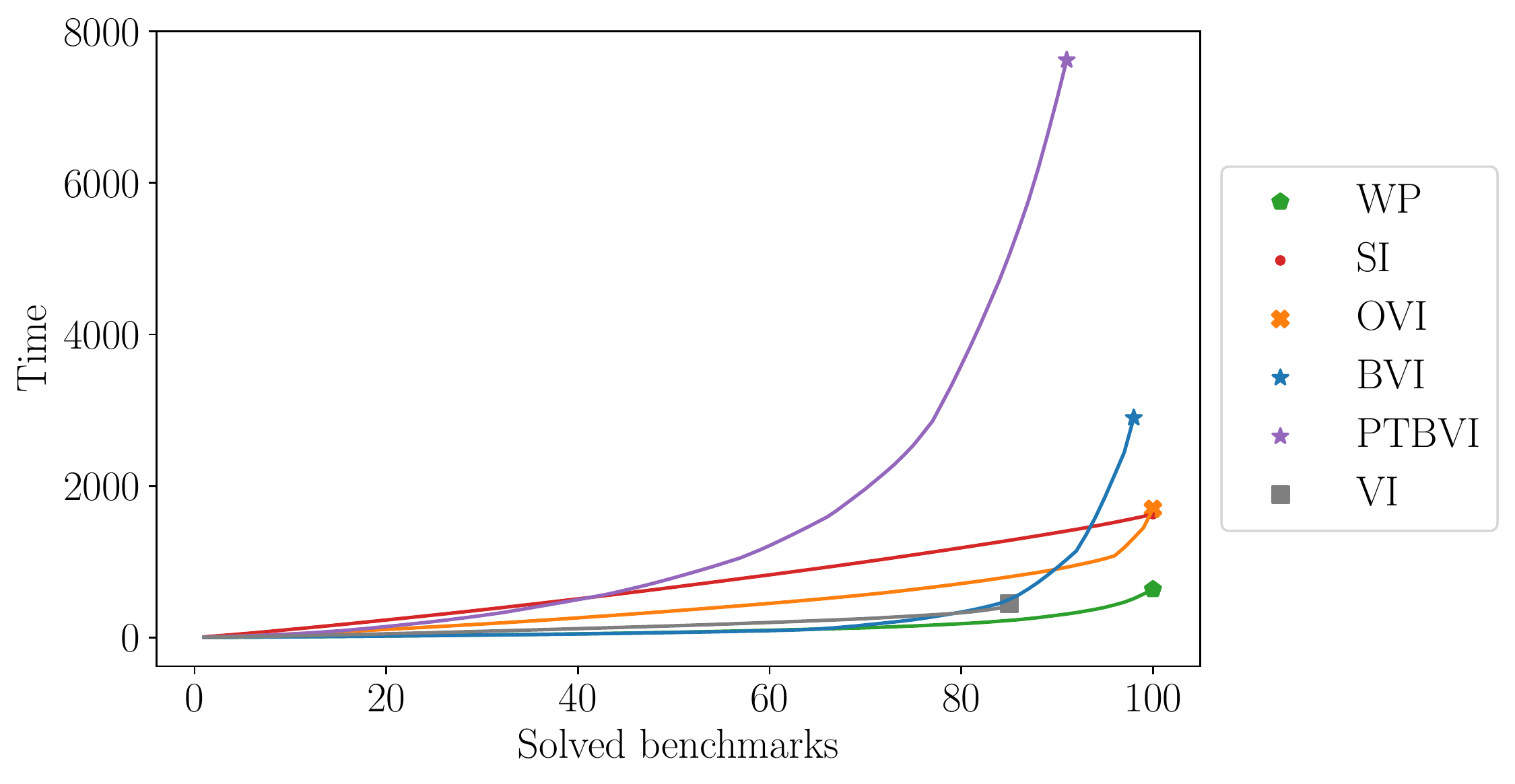} }}%
	\caption{Overview of the performance of the main algorithms on the \textit{real} and \textit{random} case studies. See Section~\ref{sec:exp:overview} for a description.}%
	\label{fig:BenchTime}%
\end{figure}

Figure~\ref{fig:BenchTime} gives an overview of the performance of the algorithms on the real and random case studies. 
The plots depict the number of solved benchmarks (horizontal axis) and the time it took to solve them (vertical axis). For each algorithm, the benchmarks are sorted in ascending order by verification time. A line stops when no further benchmarks could be solved.
Intuitively, the further to the bottom right a line extends, the better.
The algorithms shown in the legend on the right are sorted based on their performance, in descending order.
Note that these plots have to be interpreted with care, as they greatly depend on the selection of benchmarks. 

The precise algorithms provide harder guarantees, so we expect them to be slower. 
This is visible for $\PTBVI$, which is slower and solves less benchmark than others.
Still, $\PTBVI$ is optimal on certain kinds of models, as we detail in Section~\ref{sec:exp:precise-details}.
Surprisingly, $\SI$ performed very well, even competing with the approximate algorithms $\BVI$, $\OVI$ and $\WP$. 
However, this comes from the model selection, particularly of the random models. Firstly, they exhibit very small transition probabilities, since we wanted the models to be hard for VI so that we can distinguish the different stopping criteria. This slows down convergence of VI, but does not affect SI. 
Secondly, they contain few states, so using a linear program is feasible.
In Appendix~\ref{app:exp:pt}, we show that as model size increases, $\SI$ becomes less viable.

The algorithms giving $\varepsilon$-guarantees are overall quite comparable.
This was also the case in the evaluation of~\cite{ovi}, where the authors note that \emph{``for probabilistic reachability, there is no clear winner''}.
In Section~\ref{sec:exp:approx-details}, we give more details on how the performance of certain algorithms is affected by the structural features of a case study.
Note that we included classical \VI\ as a baseline, even though it gives no guarantees. It returned wrong results on two random models as well as the handcrafted haddad-monmege and MulMec.

Finally, it is important to note that random models of size 10,000 were already very hard for all algorithms, while some real models with more than 100,000 states could be solved quickly. This confirms the hypothesis of~\cite{gandalf20} that the graph structure of an SG (e.g.,\ number of actions per state, depth of topological ordering, connectedness) is more important than its pure size.

\subsection{Detailed analysis of precise algorithms.}\label{sec:exp:precise-details}
\PTVI\ and \PTBVI\ are able to solve the chain of SCCs MulMec where normal topological VI~\cite{gandalf20} was stuck, so we achieved our original goal.

We use scatter plots to evaluate the algorithms' performance in detail. Each point in a scatter plot denotes a model. 
If a point is below the diagonal, the algorithm on the horizontal axis required more time to solve it than the corresponding algorithm on the vertical axis and vice versa.
The two lines next to the diagonal mark the case where one algorithm was twice as fast as the other.

Figure~\ref{fig:PTBVIvsTLPSIvsWP} shows a scatter plot of \PTBVI\ (which performed better than \PTVI) versus the precise \SI\ and the approximate, but very performant \WP.
While on smaller models \PTBVI\ does not perform very well (Figure~\ref{fig:PTBVIvsTLPSIvsWP}(a)), on larger models it often outperforms $\SI$,  
in many cases halving the runtime or even reducing it by an order of magnitude, as shown in Figure~\ref{fig:PTBVIvsTLPSIvsWP}(b). 
We conjecture that this comes from the fact that \SI\ has to solve a linear program multiple times, while \PTBVI\ only guesses the optimal strategies once and then solves a single Markov chain.
We emphasize that \PTBVI\ never had to resort to actually performing strategy iteration, because it guessed the correct strategies in all case studies.
Moreover, \PTBVI\ even beats the best approximate method, \WP, in sufficiently large instances that contain multiple chained SCCs.
In summary, \PTBVI\ is a promising alternative to \SI\ when needing precise solutions, especially on large models with chains of SCCs.

\begin{figure}[t]
	\centering
	\subfloat[][\centering
	Comparison on all models except large
	]{\includegraphics[width=.45\textwidth]{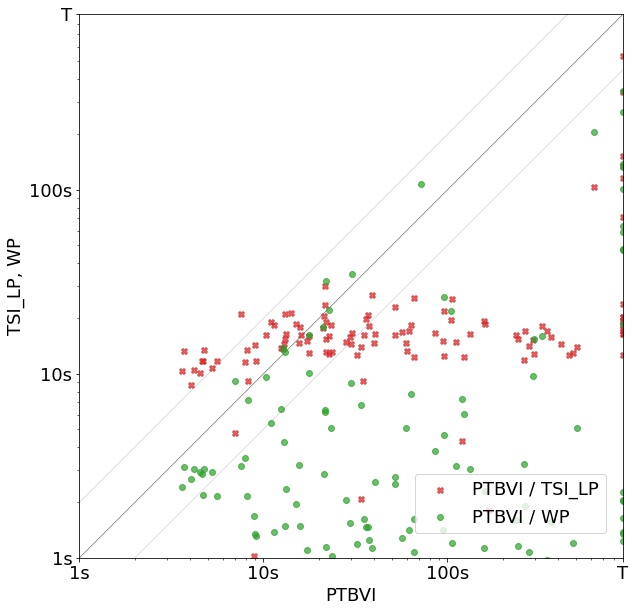} }%
	\
	\subfloat[
	Comparison on large models]{\includegraphics[width=.45\textwidth]{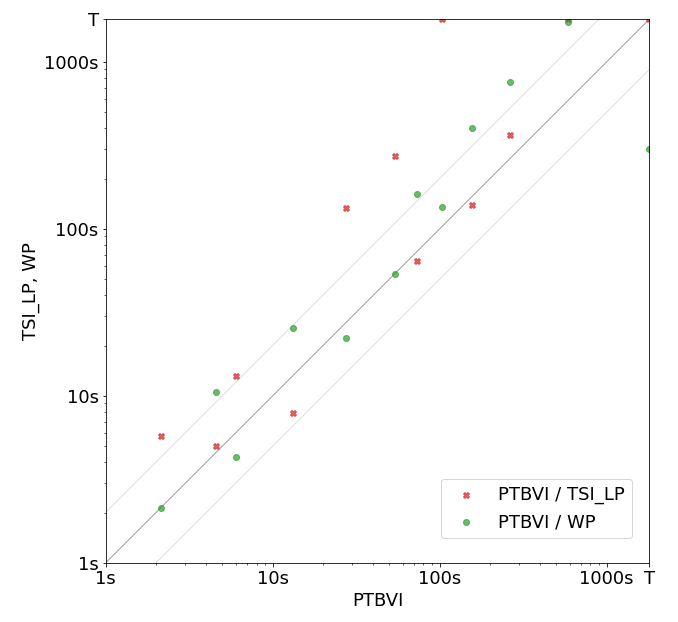} }%
	\caption{\PTBVI\ compared to $\SI$ and \WP on all datasets.}
	\label{fig:PTBVIvsTLPSIvsWP}%
\end{figure}

\subsection{Detailed analysis of approximate ($\varepsilon$-precise) algorithms}\label{sec:exp:approx-details}

All $\varepsilon$-precise algorithms perform similarly well.
\WP\ has the smallest accumulated runtime (Figure~\ref{fig:BenchTime}), no models where it is significantly worse than \BVI\ (Appendix~\ref{app:exp:bvi}) and only few models where it is significantly worse than \OVI\ (Figure~\ref{fig:OVIvsWPvsBVI}). 
As already observed in~\cite{widest}, it is particularly good when there are several or many MECs (especially on the handcrafted MulMec).
Thus, it is a valid initial choice except when the models are large with a chain of big SCCs, where we concluded in Section~\ref{sec:exp:precise-details} that \PTBVI\ is better.

We analysed \OVI\ in more detail to find our what features of the model affect its performance.
Details validating the following statements are provided in Appendix~\ref{app:exp:add}.
Intuitively, \OVI\ outperforms the other algorithms when the lower bound quickly converges, but the upper bound does not.
Dually, if the lower bound converges slowly, this is problematic for \OVI.
Note that there are many hyper-parameters of \OVI, for example the number of steps in the verification phase or the modification of the precision after a failed verification phase. 
We conjecture that these parameters affect the runtime and the choice can be improved; however, it is unlikely that there are parameter choices suitable for all kinds of models.

\begin{figure}[t]
    \centering
	\subfloat[
	Real and Handcrafted models ]{\includegraphics[width=.45\textwidth]{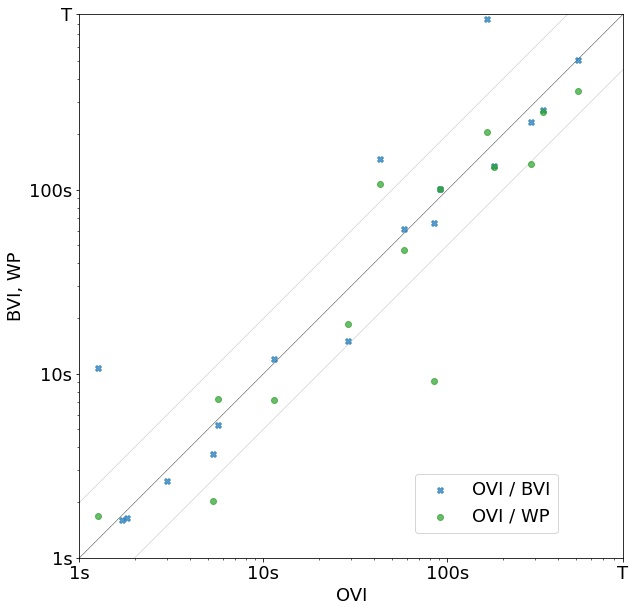} }%
    \
    \subfloat[
    Random models]{\includegraphics[width=.45\textwidth]{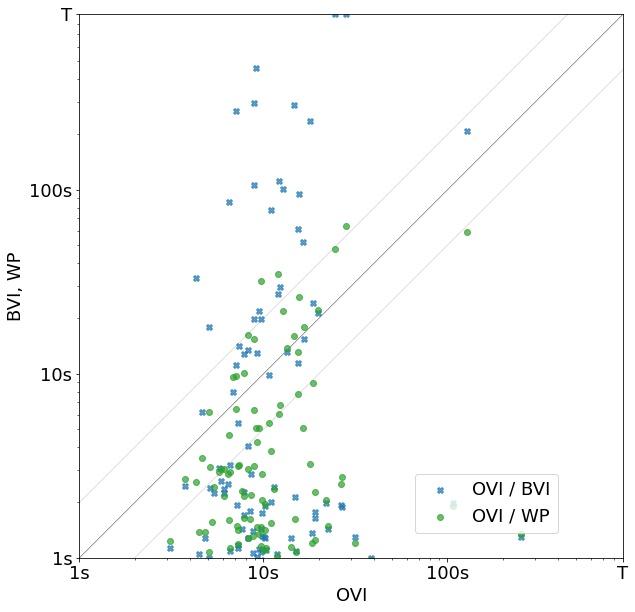} }%
	\caption{OVI compared to BVI and WP on all datasets.}
    \label{fig:OVIvsWPvsBVI}%
    \end{figure}
\section{Conclusion}
\label{sec:conc}

We extended optimistic VI from MDPs to SSGs. 
Moreover, using the ``optimistic'' idea, we fixed the issue of topological VI, so that it works even in the case of deeper models with more SCCs arranged in longer chains in the topological order.
Besides, this fix also makes the method return the exact result.
While this may be at the cost of a higher runtime, it becomes the only option when the overall model is very large, so that per-SCC analysis becomes unavoidable, and deep, so that precise values must be computed to converge at all. PTVI can be viewed as a separate algorithm or as an optimization on top of any approach from which a strategy can be extracted.

The experimental results show that the algorithms are of comparable performance, especially on real models from the standard benchmark sets.
However, an in-depth analysis of the handcrafted and random models reveals that the performance of these algorithms is often sensitive to the
underlying graph structure and, thus, their performance can vary accordingly.
While we discuss some rules of thumb as to which algorithm is to be used for a particular benchmark, a part of the future work is to provide clearer and more algorithmic recommendations.
An interesting direction here might also be to apply machine learning to recommend the most appropriate algorithm, as done for software model checkers already a few years ago, e.g., \cite{DBLP:conf/sigsoft/CzechHJW17}.

Moreover, we introduced a random generator, capable of producing various patterns even to extreme degrees. While this is very useful to find bugs and corner cases, many of the patterns need not be realistic. Consequently, we introduce a powerful set of tools to bias the generation. Nevertheless, future work shall amend this spectrum of tools with further hyper-parameters and approaches. We hope to hereby establish the platform for the community to contribute, complementary to benchmark sets \cite{PRISMben,DBLP:conf/tacas/HartmannsKPQR19}.
\bibliographystyle{splncs04}
\bibliography{ref}
\begin{appendix}
   \newpage

\section{Proof of Theorem~\ref{thm:ovi}}\label{app:ovi}

We first give an overview of the proof and delegate the technical details to the proofs of the lemmata below.

\textbf{Correctness:}	
	Since $\lb$ is computed by a classic VI variant, we know that $\lb \leq \val$ by e.g., \cite{visurvey}. We have guessed $\ub$ such that $\diff(\ub(s),\lb(s)) \leq \varepsilon$ for all $s \in \states$.
	It remains to show that upon termination $\val \leq \ub$, i.e.\ that $\ub$ is a correct upper bound.
	
	The algorithm terminates if and only if it has found an upper bound with $\ub' = \BopD_\lb(\ub) \leq \ub$ (see Lines~\ref{line:ovi:verifTermCondition} and~\ref{line:ovi:return}).
	To show that this implies our goal $\ub \geq \val$, we apply standard arguments from lattice theory~\cite[Theorem 8.20]{lattice-book}\footnote{This is what was called \emph{Park induction} in~\cite{ovi}. Since we were not able to access the work cited in that paper, we use a variant of the same claim from another textbook.}. This is the point of Lemma~\ref{lem:ovi_corr}.
	For it, we need that $\BopD$ is a monotonic operator, which is shown in Lemma~\ref{lem:ovi-bop-monotonic}.
	
	\textbf{Termination:} 
	We prove this in Lemma~\ref{lem:ovi-term} by a case distinction on the relation between $\ub$ and $\val$, similar to~\cite{ovi}, but aggregating some cases and including an argument about the additional deflating.

\begin{lemma}\label{lem:ovi-bop-monotonic}
	$\BopD_\lb$ is monotonic, i.e.\ $f_1 \leq f_2$ implies $\BopD_\lb(f_1) \leq \BopD_\lb(f_2)$.
\end{lemma}
Before we prove this, we mention this is not true when using $\BopD$ as BVI from~\cite{KKKW18} does.
There, the $\lb$ used for guessing the SECs is always updated and changed in every iteration. This can theoretically lead to non-monotonic behaviour.
Thus, it is important that we fix $\lb$ (and the implied set of SEC-candidates) for the whole verification phase. 
For Algorithm~\ref{alg:ovi}, this is obviously true as we do not update $\lb$ during the verification phase. 
However, we must not continue iterating the lower bound during verifications as in~\cite{ovi} or depend on the upper bound when guessing the SECs.

\begin{proof}
	$\BopD_\lb(f)$ modifies the given function two times: 
	first by applying the Bellman operator $\Bop$, which is monotonic, as addition, multiplication, taking maximum or minimum and their combinations are monotonic.
	Secondly, $\BopD_\lb(f)$ applies deflating. Since $\lb$ is fixed, the set of states affected by this is constant, namely the union of all SEC-candidates according to $\lb$. For every state in this set, $f$ is updated to $\min(f(s),\bestExit_f(T))$. The $\min$ ensures that the estimate cannot increase, and the computation of the best exit again contains only $\max$, multiplication and summation.
	Thus, $\BopD_\lb$ is a monotonic operator.
\end{proof}

\begin{lemma}\label{lem:ovi_corr}
	$\BopD_\lb(\ub) \leq \ub$ implies that $\val \leq \ub$.
\end{lemma}
\begin{proof}
	We restate \cite[Theorem 8.20]{lattice-book}: Given a complete lattice $(P,\leq)$ and a monotonic self-map $f$ on this lattice, $f(x) \leq x$ implies $\lfp(f) \leq x$, where $\lfp(f)$ denotes the least fixpoint of $f$.
	
	We consider the complete lattice $(\{\ub \mid \ub: \states \to \Reals\},\leq)$ of functions mapping $\states$ to probabilities with the ordering according to pointwise comparison.
	The bottom element of the lattice is the vector that maps every state to $0$, the top element maps every state to $1$.
	
	$\BopD_\lb$ is a self-map on this lattice, as it takes a function of values for every state and returns another such function.
	Moreover, it is monotonic by Lemma~\ref{lem:ovi-bop-monotonic}.
	So, using~\cite[Theorem 8.20]{lattice-book} and instantiating $f$ with $\BopD$ and $x$ with $\ub$, we get $\BopD(\ub) \leq \ub$ implies $\lfp(\BopD) \leq \ub$.
	Thus, to prove our goal, it remains to show that $\lfp(\BopD) = \val$.

	Deflating is sound and monotonic by \cite[Lemma 3]{KKKW18}, i.e.\ it can never decrease a function $\ub$ with $\ub \geq \val$ below $\val$ and it can never increase a function. 
	So deflating any $x \leq \val$ just returns $x$ and thus for $x \leq \val$, $\BopD$ does the same thing as $\Bop$.
	From this and using that $\val$ is the least fixpoint of the Bellman operator $\Bop$, see e.g.,~\cite{visurvey}, we conclude 
	\[\val = \lfp(\Bop) = \lfp(\BopD_\lb).\]

\end{proof}

\begin{lemma}\label{lem:ovi-term}
	Given an SG $\GG$ and a lower bound $\lb_0 \leq \val$, OVI$(\GG,\lb_0,\varepsilon,\varepsilon)$ terminates.
\end{lemma}
\begin{proof}
	We prove this by a case distinction on the relation between $\ub$ and $\val$, similar to~\cite{ovi}. 
	Our proof mainly differs in that we aggregate Cases 2 and 3 in their proof into case 1 of our proof and include an argument about the additional deflating.
	\begin{itemize}
		\item $\ub(s) \geq \val(s)$ for all $s \in \states$:
		If $\ub$ truly is an upper bound on the value, we want to terminate. Note that this is also the only case in which we can terminate, as otherwise we recursively call OVI again with more precision for VI and more iterations in the verification phase.
		Thus, since the precision of VI increases, we can assume that $\lb$ has converged close enough to $\val$ such that all SECs are detected correctly.
		Note that it is possible that we terminate even if the SECs are not yet guessed correctly. However, to prove termination, this assumption is necessary.
				
		By \cite[Theorem 2]{KKKW18}, the value function $\val$ is the unique fixpoint of the operator $\BopD_\lb$, since Bellman updates and deflating together ensure that $\ub$ converges to the value.
		
		As in~\cite{ovi}, we made the Bellman update monotonic by using Line~\ref{line:ovi:updateUmonotonic}, to avoid a particular corner case where there is an alternating increase and decrease, see~\cite[Section 4.2]{ovi}.
		
		The remaining argument is the same as in Cases 2 and 3 of the original termination proof.
		Intuitively, given enough steps (and the number of steps in the verification phase increases, so there eventually will be enough steps), all states see a decrease in value.
		The corner case of a state already having the correct value and hence not decreasing is handled in the same pragmatic way: we just abort the verification phase and try again with another upper bound that is necessarily higher than before.
		
		\item Otherwise, so if there exists a state $s \in \states$ with $\ub(s) < \val(s)$:
		In such a case, the verification phase will certainly be aborted, as we terminate if and only if $\BopD_\lb(\ub) \leq \ub$, which cannot happen when some $\ub(s) < \val(s) = \lfp(\BopD_\lb)(s)$.
		Thus, we need to abort the verification phase with this $\ub$, which happens after a finite number of iterations (see Line~\ref{line:ovi:forVerifPhase}). 
		Note that the additional breaking conditions mentioned at the end of Section~\ref{sec:ovi} help to detect this case faster.
		
		It remains to show that we cannot stay in this case forever, but that eventually we guess a $\ub \geq \val$.
		This happens, because $\lb$ converges to $\val$, see e.g.,~\cite{visurvey}.
		Since the increase of $\diff_\varepsilon^+$ for guessing the upper bound is constant\footnote{Or, in the case of relative error, can be lower bounded by $\lb(s) * \varepsilon$}, we will eventually guess an upper bound that is greater than the value and our algorithm will eventually detect that and terminate.
	\end{itemize}
\end{proof}
\section{Proofs for Section~\ref{sec:pt}}\label{app:pt}

\subsection{Details on PTVI (Algorithm~\ref{alg:pt})}\label{app:pt:details}

\paragraph{Computing optimal strategies:}

How to compute the strategies is a heuristic that should get the best from the given information of the VI-algorithm.
We phrased it such that the computation depends on both a lower bound $\lb$ and an upper bound $\ub$ in order to show how to use it with both BVI and OVI. When using naive VI, we set $\ub \eqdef \lb$.
We use $\ub$ to guess the Minimizer strategy and $\lb$ to guess the Maximizer strategy, as this is using the most conservative estimate that we have available. 

From an estimate function, we can derive a strategy by picking actions that are optimal according to this estimate, similar to the Bellman equations~\ref{eq:bellman}.
Indeed, for Minimizer this is sufficient, so we set
\[
	\tau \gets \{(s,a) \mid s \in \minStates, a \text{ picked randomly from } \argmin_{a' \in \act(s)} \ub(s,a')\}
\]
For Maximizer this is not sufficient, because Maximizer might have actions with optimal value that however do not make progress towards the target. Imagine for example a Maximizer state with an action $a$ that loops and an action $b$ leading to the goal.
While picking $a$ for a finite number of times does not decrease the value, as the state can still play $b$ and reach the goal surely, committing to playing $a$ infinitely often reduces the value to 0. 
Thus, when deriving a Maximizer strategy, we cannot just pick some optimal action.

There are several ways to deal with this problem:
firstly, to obtain a memoryless deterministic strategy, one can use graph algorithms to compute a distance measure between states in $\states_?$ and goal states; then the strategy picks actions that not only are optimal according to the estimate, but additionally reduce the distance to the goal.
This is similar to the construction used in \cite[Theorem 2]{DBLP:conf/isaac/AnderssonM09} for deriving strategies of Maximizer in a reachability game.
However, we want to avoid the additional computation time for these graph algorithms. 
Thus, we use a second approach: strengthening the notion of strategy to allow for randomization. 
A randomized strategy is map $\states \to \Dist(\act)$ that for every state $s$ returns a probability distribution over available actions $\Av(s)$. We write $(s,a,p) \in \sigma$ to say that $\sigma(s)(a) = p$; a strategy is well-defined when all triples with $p>0$ are enumerated.
Thus, we select the Maximizer strategy as
\[
\sigma \gets \{(s,a,p) \mid s \in \maxStates, a \in \argmax_{a' \in \act(s)} \lb(s,a'), p = \frac{1}{ \abs{\argmax_{a' \in \act(s)} \lb(s,a')}} \}
\]
This is correct, because every action that is optimal (according to the estimate) has positive probability to be selected. The ``staying'' actions do not decrease the value, and almost surely we eventually pick an action that makes progress towards the goal.
In the words of \cite[Lemma 5]{DBLP:conf/isaac/AnderssonM09}: the strategy is safe and stopping.
Note that this does not drastically decrease the transition probabilities in the resulting Markov chain, as there typically are less than 3 actions per state. Hence it poses no problem for the Markov chain solving.

\paragraph{Markov chain solving:}
To solve the induced Markov Chains, we use the standard equation approach as described in \cite[Chapter 11]{introProb}. 
First, the MC is represented in a transition matrix. Next, one has to find the states $s \in \states^?$ which cannot reach any target in the MC and remove their rows and columns from the matrix. 
To do so, it is usually necessary to perform an all-pairs-shortest-path algorithm like Flyord-Warshall. However, since this requires $O(|\states|^3)$ operations, and we already have $\lb$ from the VI process, we can obtain said states instead directly by checking whether $\lb(s) = 0$. 
Removing states that cannot reach any target from the transition matrix guarantees non-singularity, and thus invertability. Lastly, the part of the matrix corresponding to $\states^?$ needs to be inverted and multiplied with the part of the transition matrix corresponding to $\targets$, yielding the reachability probability for every state $s \in \states^?$ in the Markov chain.

In our implementation, we use the JAMA library\footnote{\url{https://math.nist.gov/javanumerics/jama/}} to perform matrix operations like multiplications or inversions.

\paragraph{Strategy iteration in the case of non-optimal strategies:}
The local check we perform in Line~\ref{line:pt:if} gives guarantees if and only if it succeeds. So if it fails, we know nothing about the strategies and have to perform normal SI. 
This means we can only fix one of them, as solving the resulting MDP might return a different strategy for the other player.
However, it is quite likely that both strategies are good.
Note that for SI we need memoryless deterministic strategies, so using the randomized $\sigma$ as we defined above is not possible. 
Thus it makes sense to resort to starting with $\tau$.

\subsection{Proof of Theorem~\ref{thm:pt}}\label{app:pt:proof}

\setcounter{theorem}{1}

\begin{theorem}\label{thm:pt-app}
	Algorithm~\ref{alg:pt} returns the precise solution $\val$.
\end{theorem}
\begin{proof}
	The argument for correctness of solving the SG in topological order is the same as in~\cite[Section 4.4]{gandalf20}.
		
	In the case that the condition in Line~\ref{line:pt:if} evaluates to false, there is nothing to prove, because in this case the algorithm falls back to using strategy iteration which is known to return the precise solution, see e.g.~\cite{DBLP:conf/dimacs/Condon90}.
	So we only have to prove that, if the check in Line~\ref{line:pt:if} evaluates to true, this implies that $\sigma$ and $\tau$ are optimal strategies and $\val_\Gst$ is the value function of the SG.
	
	In the following, for any SG (or MDP or Markov chain) $\GG$, we use $\val_\GG$ to denote its value function.
	Further, $\Gsigma$ and $\Gtau$ are the MDP with Maximizer strategy $\sigma$ respectively Minimizer strategy $\tau$ fixed, and $\Gst$ is the Markov chain with both strategies fixed.
	
	Intuitively, our local check amounts to performing SI on the MDPs $\Gsigma$ and $\Gtau$.
	For the proof, first consider the MDP $\Gsigma$. SI in this MDP starts by fixing an arbitrary Minimizer strategy and we choose $\tau$.
	Then it computes the value of the resulting Markov chain $\Gst$ and checks whether Minimizer wants to switch the strategy in any state, i.e.\ whether there exists an $s \in \minStates$ such that $\tau(s) \notin \argmin_{a \in \act(s)} \val_{\Gst}(s,a)$. If Minimizer wants to switch, $\tau$ is not optimal and the check evaluates to false.
	Otherwise, $\tau$ is an optimal strategy in the MDP $\Gsigma$, cf.~\cite{Puterman}.
	Thus, we have $\val_{\Gsigma} = \val_{\Gst}$.
	
	Dually, the check for the Maximizer strategy verifies that $\val_{\Gtau} = \val_{\Gst}$ by proving that $\sigma$ is an optimal strategy in the MDP $\Gtau$.
	Note here that, since we use a randomized $\sigma$, we have to check that every action to which $\sigma$ assigns positive probability is optimal, i.e. $\{a \mid \exists p>0: (s,a,p) \in \sigma\} \subseteq \argmax_{a \in \act(s)} \val_{\Gst}(s,a)$.
	
	Finally, we have to relate $\val_{\Gsigma}$ and $\val_{\Gtau}$ to the value on the actual game $\val_{\GG}$.
	Observe that by fixing an arbitrary Maximizer strategy, we can only decrease the value, as it might be suboptimal, but we can never increase the value. Formally, $\val_{\Gsigma} \leq \val_{\GG}$.
	Dually, $\val_{\GG} \leq \val_{\Gtau}$.
	Putting everything together, we have 
	\[
	\val_{\Gst} = \val_{\Gsigma} \leq \val_{\GG} \leq \val_{\Gtau} = \val_{\Gst}.
	\]
	This proves that indeed both strategies are optimal in $\GG$ and $\val_{\Gst}$ is the correct value function.
	
\end{proof}
\section{Additional details for randomly generated models} \label{app:randomGen}

\subsection{Our algorithm for random model generation} \label{sec:randomGenAlgo}

We use Algorithm \ref{alg:randomRandom} to create any random stochastic game that is connected from the initial state.
After initialization, the algorithm has two phases: the forward and the backward procedure.
During initialization, we generate a random number $n$ and create a set of states $\states := [n]$. Also, we assign states at random to either Maximizer or Minimizer.
In the forward procedure, we iterate over every state $\state \in \states$ and make sure that a previous state $\state'$ is connected to it by providing an action with positive transition probability to $\state$ from $\state'$.
This guarantees that the initial state can reach every state in the stochastic game.
The backward procedure then adds arbitrary actions to arbitrary states to enable generating every possible SG.

To generate the actions of a state, we use Algorithm \ref{alg:FillActions}. 
It receives a state-action pair $(\state, \action)$ where $\sum_{\state' \in \states} \trans(\state, \action, \state') < 1$.
It then increases the transition probability of a randomly selected state $\state' \in \states$ where $\trans(\state, \action, \state') = 0$. 
This is repeated until $\sum_{\state' \in \states} \trans(\state, \action, \state') >= 1$ or there is no state $\state'$ that is not yet reached by $(\state, \action)$.
In case $\sum_{\state' \in \states} \trans(\state, \action, \state') < 1$ holds but the state-action pair is reaching every state in $\states$, 
we increase the most recently increased transition probability such that the resulting distribution is a probability distribution.
If we reach $\sum_{\state' \in \states} \trans(\state, \action, \state') > 1$, we reduce the transition probability we increased most recently.
After applying Algorithm \ref{alg:FillActions}, $(\state, \action)$ is a valid transition distribution, i.e.\ its probabilities sum up to~1.

\begin{algorithm}[ht]
    \caption{Generating random models connected from initial state}
    \label{alg:randomRandom}
    \begin{algorithmic}[1]
    \Ensure Stochastic game $\GG$ where the initial state is connected to any $\state \in \states$
    \State Create $\states$ with a random $n \in \Naturals$
    \State Partition $\states$ uniformly at random into $\maxStates$ and $\minStates$
    \State Enumerate $\state \in \states$ in any random order from 0 to n-1
    \State Set $\initstate$ to the state with index 0
    \State Set $\targets = \{s_{n-1}\}$
    \For{$\state = 1 \rightarrow n-1$} \Comment{Forward Procedure}
        \If{$\state$ does have an incoming transition} 
            Continue (Skip iteration)
        \Else
            \State Pick any state $\state'$ with index smaller than $\state$
            \State Create an action $\action$ that starts at $\state'$ 
            \State Assign to ($\state'$, $\action$) a positive probability of reaching $\state$
            \State Create a valid probability distribution for ($\state'$, $\action$) by applying FillAction($\state'$, $\action$)
            \State Add $\action$ to $\Av(\state')$
            \State Add $\action$ to $\actions$
        \EndIf
    \EndFor
    \For{$\state = n - 1 \rightarrow 0$} \Comment{Backward Procedure}
        \State Pick a random number $m \in \Naturals$ \Comment{Add as many actions as possible}
        \If{$|\Av(\state)| = 0$} $m \gets \max{\{m, 1\}}$ \Comment{Every state must have at least one action} \EndIf 
        \For{$i = 1 \rightarrow m$}
            \State $(\state, \action_i)$ = FillAction($\state$, $\action_i$)
            \State Add $\action_i$ to $\Av(\state)$
            \State Add $\action$ to $\actions$
        \EndFor
    \EndFor
    \end{algorithmic}
\end{algorithm}

\begin{algorithm}[ht]
    \caption{FillAction($\state$, $\action$)}
    \label{alg:FillActions}
    \begin{algorithmic}[1]
        \Require outgoing state $\state$, action $\action$
        \Ensure action $\action$ that has a valid underlying transition probability distribution
        \Repeat
            \State Pick a random state $\state'$ where $\trans(\state, \action, \state') = 0$
            \State Increase $\trans(\state, \action, \state')$ by a random number $\in (0, 1]$
            \Until{either $\sum_{\state \in \states} \trans(\state, \action, \state') \geq 1$ or $\forall \state \in \states: \trans(\state, \action, \state') > 0$}
        \If{$\sum_{\state \in \states} \trans(\state, \action, \state') > 1$}
            \State Decrease the most recently modified $\trans(\state, \action, \state')$ so $\sum_{\state \in \states} \trans(\state, \action, \state') = 1$
        \ElsIf {$\sum_{\state \in \states} \trans(\state, \action, \state') < 1$} \Comment{Loop terminated because $\forall \state \in \states: \trans(\state, \action, \state') > 0$}
            \State Increase the most recently modified $\trans(\state, \action, \state')$ so $\sum_{\state \in \states} \trans(\state, \action, \state') = 1$
        \EndIf
    \Return ($\state$, $\action$)
    \end{algorithmic}
\end{algorithm}

\begin{lemma}
    Algorithm \ref{alg:randomRandom} creates formally correct stochastic games.
\end{lemma}
\begin{proof}
    $\states$ is finite since its size is determined by a random number $n \in \Naturals$ (Line 1).
    Line 2 ensures that we have a partition of $\states$ into $\maxStates$ and $\minStates$.
    Next, Lines 3-5 provide an initial state and a target.
    Thus, we only need to argue that $\Av$ is truly a mapping of
    $\states \rightarrow 2^{\actions}$ and that the transition function yields a probability distribution.
    When we introduce state-action pairs (Lines 9-14 and Lines 18-21), we introduce a new action that we add to $\actions$. 
    Also, we add the state-action pair to $\Av$ (Lines 12 and 19).
    Thus, any $\Av$ is function of $\states \rightarrow 2^{\actions}$.

    Next we need to prove that for every $\state \in \states$ and every action $\action \in \Av(\state)$ the transition function $\trans(\state, \action)$ yields a probability distribution.
    In other words we need to validate that 
    (i) for every state $\state' \in \states: \trans(\state, \action, \state') \in [0, 1]$ and
    (ii) $\sum_{\state' \in \states} \trans(\state, \action, \state') = 1$ are true. 

    To prove (i), note that whenever we introduce an action $\action$ to the set of enabled actions $\Av(\state)$ of a state $\state \in \states$, 
    we have $\trans(\state, \action, \state') = 0$ for all $\state' \in \states$.
    We increase $\trans(\state, \action, \state')$ in Line 11 of Algorithm \ref{alg:randomRandom} and Lines 3 and 8 of Algorithm \ref{alg:FillActions}.
    We increase transition probabilities by numbers in $[0,1]$.
    Due to the condition in Line 4 of Algorithm \ref{alg:FillActions}, $\trans(\state, \action, \state')$ can only be increased a second time
    in Line 8 of Algorithm \ref{alg:FillActions}. However, per state-action pair $(\state, \action$) with $\action \in \Av(\state)$ Line 8 of Algortihm \ref{alg:FillActions}
    may be executed only once and we increase only by a $\Delta > 0$ such that 
    $\trans(\state, \action, \state') + \Delta + \sum_{\stateMac'' \in \states \setminus \{\state'\}} \trans(\state, \action, \stateMac'') = 1$.
    Since every state $\stateMac'' \in \states \setminus \{\state'\}$ was increased only once, $\trans(\state, \action, \state') + \Delta \leq 1$.

    To prove (ii), note that every pair $(\state, \action)$ where $\action \in \Av(\state)$ is given to Algorithm \ref{alg:FillActions}.
    Once the loop (Lines 1-4) terminates, it either holds that (a) for every state $\state' \in \states: \trans(\state, \action, \state') > 0$ or that 
    (b) $\sum_{\state' \in \states} \trans(\state, \action, \state') \geq 1$.
    If $\sum_{\state' \in \states} \trans(\state, \action, \state') = 1$ the algorithm terminates correctly.
    Hence, we consider the other two cases:
    In case (a) Line 8 increases the most recently added triple $(\state, \action, \state')$ by a $\Delta > 0$
    such that $\trans(\state, \action, \state') + \Delta + \sum_{\stateMac'' \in \states \setminus \{\state'\}} \trans(\state, \action, \stateMac'') = 1$.
    In case (b) $\sum_{\state' \in \states} \trans(\state, \action, \state') > 1$. 
    Due to the exit condition of the loop (Line 4), without the most recently increased transition $(\state, \action, \state')$ the sum of the probabilities must be
    below 1, i.e. $\sum_{\stateMac'' \in \states \setminus \{\state'\}} \trans(\state, \action, \stateMac'') < 1$. Thus, there is a $\Delta \in (0, 1), \Delta < \trans(\state, \action, \stateMac'')$ such that
    $\trans(\state, \action, \state') - \Delta + \sum_{\stateMac'' \in \states \setminus \{\state'\}} \trans(\state, \action, \stateMac'') = 1$.
    We decrease $(\state, \action, \state')$ by this $\Delta$.
    In conclusion, every pair $(\state, \action)$ that is provided to Algorithm \ref{alg:FillActions} yields a valid transition distribution where 
    $\sum_{\state' \in \states} \trans(\state, \action, \state') = 1$.
    Thus, Algorithm \ref{alg:randomRandom} generates formally correct stochastic games.
\end{proof}

To argue about the SGs that the algorithm produces, we introduce some notation.
Let $\randomRandomOutcome$ be  the set of SGs that Algorithm \ref{alg:randomRandom} can produce and $\connectedSG$ be the set of all SGs where every state is reachable from the initial state.
Note that every SG that is not an element of $\connectedSG$ can be transformed into an SG with the same value and optimal strategies by removing all unreachable states. Hence, we only care about producing SGs in $\connectedSG$.

\begin{lemma}
Algorithm~\ref{alg:randomRandom} creates all SGs where every state is reachable from the initial state, i.e.\ $\randomRandomOutcome = \connectedSG$.
\end{lemma}
\begin{proof}

We first show that $\randomRandomOutcome \subseteq \connectedSG$:

For this statement to hold, any $\SG \in \randomRandomOutcome$ must be connected from the initial state.
Proof by induction over the indices $i$ of the states along their enumeration assigned during Algorithm $\ref{alg:randomRandom}$:

$\textbf{Basis}$: $i=0$:
$\initstate$ is the initial state. The initial state can reach itself within 0 steps.

$\textbf{Hypothesis}$:
Let $i$ be arbitrary but fixed with $i \leq n-1$, where $n = |\states|$. For every $j \leq i$ it holds that $\initstate$ can reach $\state_j$.

$\textbf{Inductive Step}$: $i \gets i+1$:
Due to the forward procedure it holds that 
\[
    \exists \state_j \in \states, j < i, \action \in \Av(\state_j): \trans(\state_j, \action, \state_i) > 0
\]
However, according to our hypothesis $\initstate$ is connected to $\state_j$ and thus also to $\state_i$.

Now we show that $\connectedSG \subseteq \randomRandomOutcome$:

Pick an arbitrary but fixed stochastic game $\GG \in \connectedSG$.
Next, we show that there is a run of our algorithm that will return a stochastic game $\GG' \in \randomRandomOutcome$ where $\GG'$ is an automorphism to $\GG$.
Thus, $\GG'$ and $\GG$ are the same except for the state enumeration and the labels of the actions.

For this, we need several statements to hold at once:
\begin{enumerate}
    \item The number of states in $\GG$ and $\GG'$ is equal.
    \item The partition of $\states$ to $\maxStates$ and $\minStates$ is the same for $\GG$ and $\GG'$.
    \item $\GG$ and $\GG'$ have the same initial states and targets.
    \item All state-action pairs in $\GG$ and $\GG'$ yield the same probability distributions in $\trans$.
    \item Every state in $\GG$ and $\GG'$ has the same actions.
\end{enumerate}

We decide randomly on the number of states in Line 1 of Algorithm \ref{alg:randomRandom} and partition $\states$ into $\maxStates$ and $\minStates$ at random
in Line 2 of Algorithm \ref{alg:randomRandom}. Thus, for every $\GG \in \connectedSG$ there exists $\GG' \in \randomRandomOutcome$ that have the same number of states to which
there is an enumeration such that they are partitioned equally in both stochastic games.
Since the states can be arranged in any order, we pick the initial state of $\GG'$ such that is the same as in $\GG$.
All targets of $\GG$ can be mapped to the singular target of $\GG'$, since they only have self-loops and behave identically to $\target$ of $\GG'$.
Thus, there exists a run where Statements 1, 2, and 3 hold. 

When using Algorithm \ref{alg:FillActions} to create a probability distribution for a state-action pair $(\state, \action)$, 
we increase transition probabilities until they sum up to 1. Thus, any summation $\sum_{\state' \in \states} \trans(\state, \action, \state') = 1$ is possible.
In consequence, an action may lead into arbitrary states, have an arbitrary number of positive transition probabilities between 1 and $|\states|$, and may have arbitrary
probability distributions on the transitions as long as they sum up to 1. So out of all runs where Statements 1, 2, and 3 hold, there also must be at least one run where statement 4 holds too.

To show that Statement 5 holds in addition, note that each $\GG \in \connectedSG$ has a minimal set of state-action tuples such that the initial state is connected to every state.
Taking this set, we can perform a breadth-first search from the initial state to provide an enumeration of the states.
If we iterate over the states along this enumeration, we can reproduce each of the actions in the minimal set during the forward process.
Due to the enumeration, to each state $\state$ except for the initial state, there is a state with a smaller index $\state'$ such that $\state'$ has an action $\action$ with a positive transition
probability of reaching $\state$. Since every other transition of ($\state'$, $\action$) can lead into arbitrary states and the probability distribution of ($\state'$, $\action$) can be arbitrary, 
we can recreate the minimal set of state-action tuples in $\GG'$.

The remaining state-action pairs of $\GG$ can be added to $\GG'$ during the backward process, where every state may add arbitrarily many actions with arbitrary transition distributions.
\end{proof}

Note that although $\randomRandomOutcome = \connectedSG$, in general Algorithm \ref{alg:randomRandom} does not sample $\connectedSG$ uniformly at random.
Due to the forward procedure, states with smaller indices tend to have more actions than the states with higher ones.
Additionally, creating the transition distributions as described in Algorithm \ref{alg:FillActions} favors state-action pairs to have few transitions.
If we pick the transition probabilities between $(0, 1]$ uniformly at random, around $83,33\%$ of all actions have two or three transitions with positive probability and 
none with one transition.

\subsection{Parameters and guidelines for model construction} \label{sec:guidelines}
We implemented a constrained version of Algorithm~\ref{alg:randomRandom} from Appendix~\ref{sec:randomGenAlgo} to randomly generate models.
The real-world implementation has to be constrained due to the natural restriction that a computer cannot generate arbitrarily large stochastic games and
arbitrarily small transitions due to finite memory. Additional constraints on the real-world implementation are the pseudo-randomness while taking decisions, as well as
floating point machine precision. 
Moreover, we want to give the users some control over the properties of the resulting models like the number of states or the partitioning into Minimizer and Maximizer states.
If all model properties vary significantly, it is very hard to deduce why an algorithm performs differently on two models.
Thus, we provide several parameters which can be set to restrict the randomization. Some examples of the parameters that one can control are the number of states, number of models, smallest probability that is allowed to occur, number of transitions of an action. 

\subsection*{Limits and additional guideline options} \label{app:random:guidelinesSubsec}
Although the parameters we expose for random generation can directly influence various structural properties of the resulting models,
there are other structural properties that are hard to influence. For example, there is no direct way to affect the size and number of SCCs of a model, or to guarantee
that every state in a model has a certain number of actions when using Algorithm \ref{alg:randomRandom}. We provide some guidelines that can be followed to influence such properties.

\begin{enumerate}
	\item \textbf{RandomTree guideline.} We refer to the guideline that controls how many actions a state has as \emph{RandomTree}
because it creates a \emph{tree-like} graph structure where the initial state is the root.
Every node of the tree has $k$ actions and at most $k$ children, where $k$ is a parameter.
Every action has an assigned child to which it has a positive transition probability.
The rest of the probability distribution of the action is assigned at random.
An inner node of the tree may have less than $k$ children if adding $k$ children would exceed the requested number of states $n$ for the model.
Also, leaves are not required to have $k$ actions. Their actions are only introduced during the backward process and are there to enable the generation of end components.

	\item \textbf{RandomSCC guideline} This is a guideline that controls the SCCs of a model.
The procedure requires a minimal and maximal size boundary $[a, b]$ for every SCC and the total number of states in the stochastic game $n$.
First, we create subgames of a randomly chosen size in $[a,b]$.
The subgames are created by the algorithm described in Appendix~\ref{sec:randomGenAlgo}.
We then use Tarjan's algorithm for strongly connected components~\cite{TarjansAlgorithm} to identify the SCCs of the created subgame.
Next, we unify all the SCCs of the subgame by using the topological enumeration Tarjan's algorithm provides.
We circularly connect the SCCs along the enumeration, making the whole subgame an SCC.
Next, we make sure that the subgame is connected to the rest of the stochastic game by making sure a previously created subgame has an action
leading into this subgame. We repeat this procedure until we have at least $n$ states in the stochastic game,
resulting in a stochastic game $\GG$ that is connected from the root where the user has an easy way of controlling the number and size of the SCC in $\GG$.
\end{enumerate}
\section{Details on model analysis} \label{app:modelAnalysis}
In this section, we discuss the feature distribution of the real and randomly generated models.
We claim that such a feature analysis of benchmark sets is important in order to find out biases, as well as judge why certain sets of benchmarks are hard.
We first give a few noteworthy conclusions of our analysis.
\begin{itemize}
	\item Often large parts of the real case studies are solved by pre-computations. Hence, instead of the size of the state space $\abs{\states}$ we should consider the size of the unknown part $\abs{\states^?}$.
	\item In the real case studies we have, there are usually only few actions and successors per state, few MECs and few non-trivial SCCs. Thus, to assess general algorithm performance, we should generate models with these features, be it by hand, randomly or by finding realistic problems with this structure.
	\item In general, the models generated by our algorithm without using the guidelines are hard because they form very connected models with many actions and transitions as well as low transition probabilities. Apart from that, the models generated do not exhibit large size, many MECs or long chains of SCCs.
\end{itemize}

We provide two figures, each analysing 16 features of a set of case studies.
Figure~\ref{fig:Real_FeatureDistribution} shows the feature distribution of the real case studies. We refer to this set as \realcs.
Figure~\ref{fig:10act_FeatureDistribution} shows the feature distribution of a randomly generated set of models with many actions. We refer to this set as \many.
For our analysis we use box plots, where for every feature distribution, the plot shows the median (orange line), average (green triangle),
25 and 75 percentile (box), 10 and 90 percentile (whiskers) and outliers (circles).
In every figure, the box plots are grouped by and coloured according to the following categories:
\begin{itemize}
    \item Green outlines are for features related to states.
    \item Blue outlines are for features related to actions.
    \item Cyan outlines are for features related to transitions.
    \item Red outlines are for features related to MECs.
    \item Orange outlines are for features related to SCCs.
\end{itemize}

\begin{figure}[t]
	\centering
	\includegraphics[width=1\textwidth]{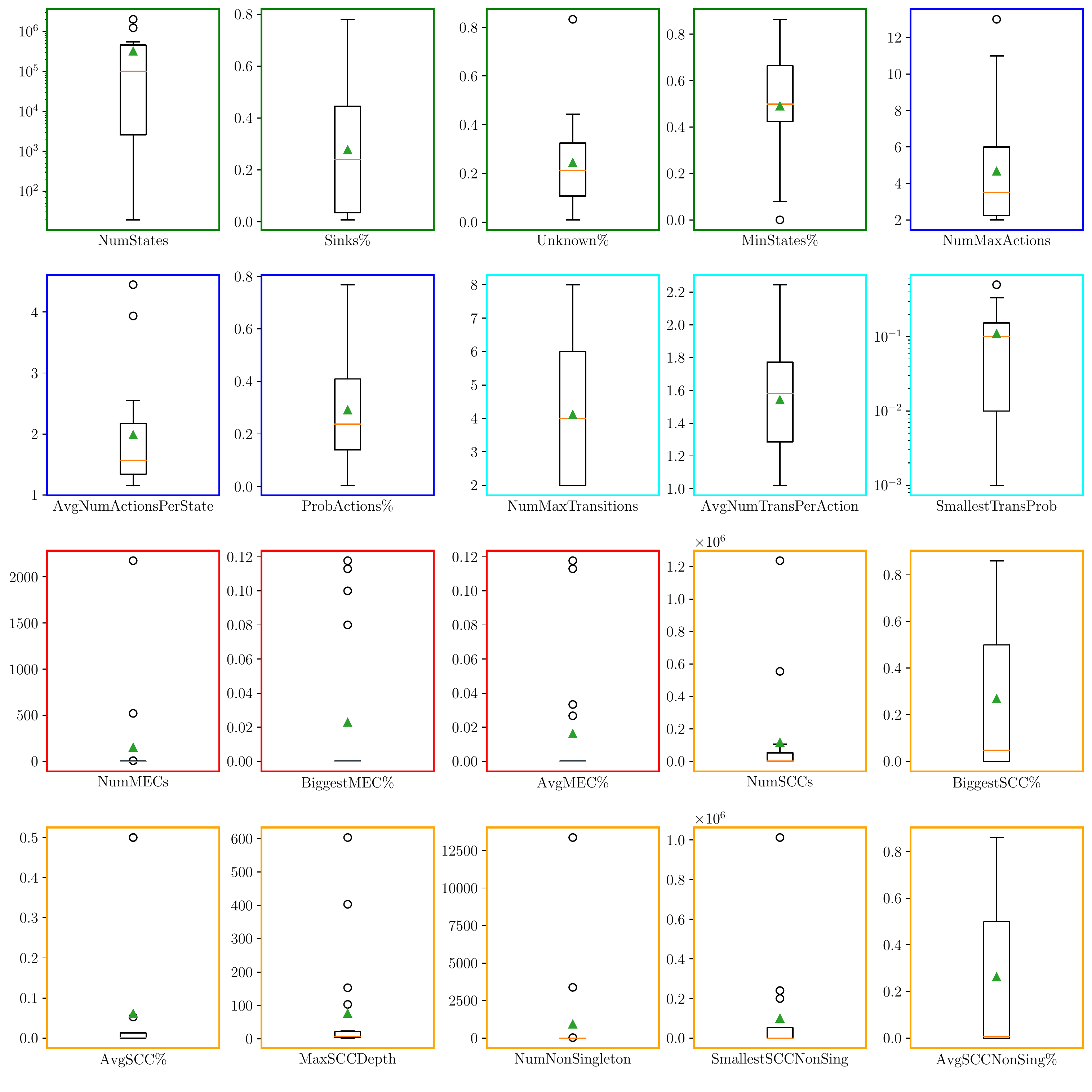}
	\caption[Feature Distribution of the case studies]{
		\realcs: Box plots for analysing the features of the \emph{real case studies}.
		A description of how to read box plots is provided in Section~\ref{sec:randomGen}.
	}
	\label{fig:Real_FeatureDistribution}
\end{figure}

\begin{figure}[t]
	\centering
	\includegraphics[width=1\textwidth]{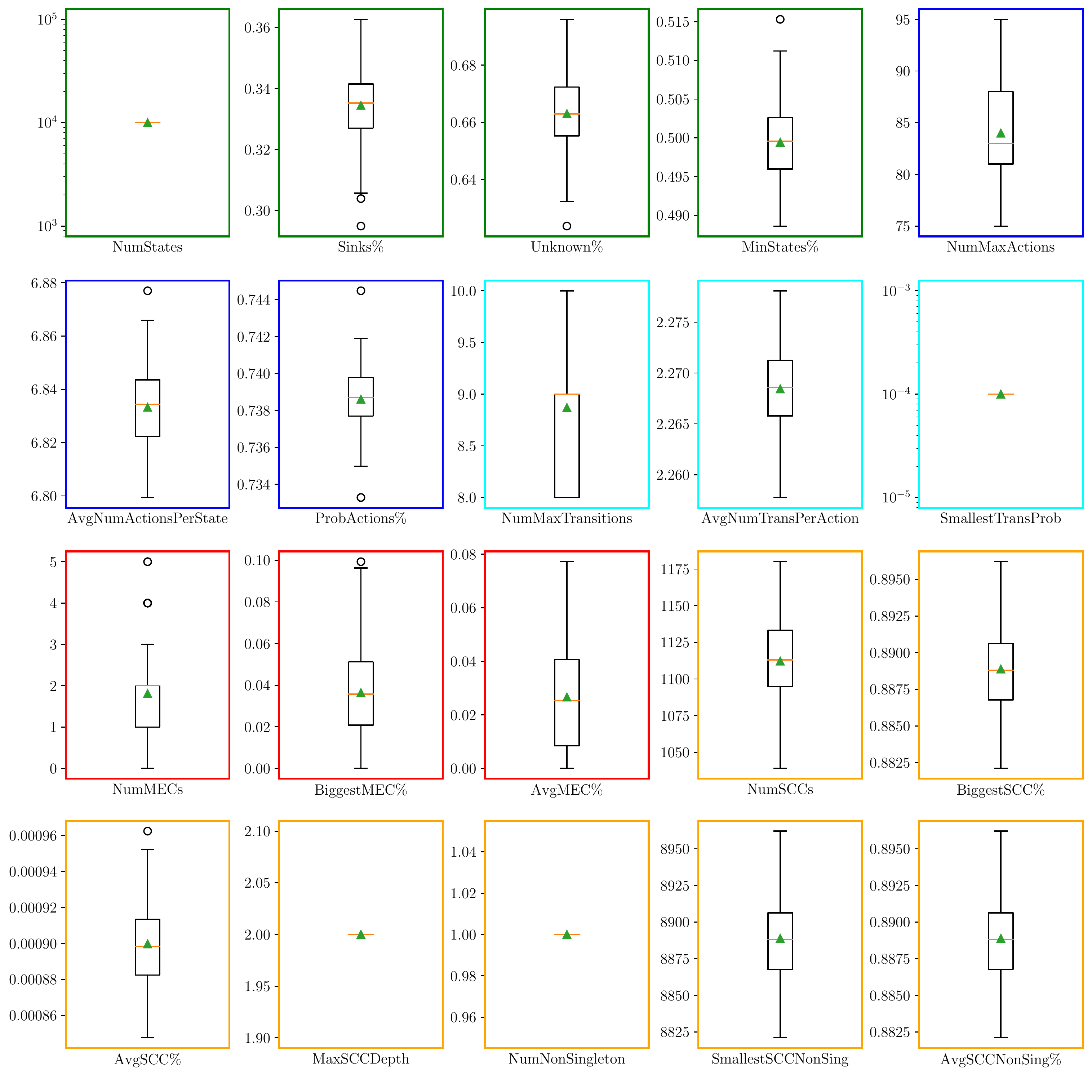}
	\caption[Feature Distribution of the case studies]{
		\many: Box plots for analysing the features of the set of \emph{randomly generated models}.
		A description of how to read box plots is provided in Section~\ref{sec:randomGen}.
	}
	\label{fig:10act_FeatureDistribution}
\end{figure}

For every feature we point out differences and similarities between the sets.
We also motivate why certain features are interesting and our choices of parameters for the set \many.
\begin{itemize}
	\item NumStates: Number of states in a model, given in log scale. 
	For \realcs, this has a spread between several dozen up to a few million states. 
	For \many, we always chose 10,000, so that the complexity of the model could not come from its size, but the size was still non-trivial.
	To analyze the complexity induced by the sheer size of a model, we used the additional handcrafted scalable model described in Appendix~\ref{app:exp:pt}.
	\item Sinks\% and Unknown\%: These features describe the percentage of states that are sinks ($\sinks$) or unknown states ($\states^?$). 
	Note that all randomly generated models have exactly one goal state.
	Hence, we have that Sinks\% + Unknown\% sums up to (a little less than) 1.
	For \realcs, the percentage of unknown states is almost always less than 40\% and less than 20\% in half the cases. This means that most of the model is solved by pre-computations, being either sinks or goal states.
	We highlight one implication of this: an experimental analysis relating the number of states to the performance of algorithms is flawed, since large parts of the state space are quickly solved by pre-computations shared by all algorithms. Thus, such an analysis should rather consider the number of unknown states.
	In \many, the percentage of unknown states is around 60\%, i.e.,\ typically more of the model requires an actual computation of the solution algorithm.
	\item MinStates\%: Percentage of states belonging to Minimizer. Note that the percentage of Maximizer states is 1 minus this. For all sets, the average is around 50\%. For \realcs, the variance is higher. 
	\item NumMaxActions, AvgNumActionsPerState, ProbActions\%, NumMaxTransitions, AvgNumTransPerAction:
	These features intuitively describe the ``breadth'' or ``branching'' of the model. They are, in order: the maximum number of actions per state occurring in the model, the average number of actions per state, the percentage of probabilistic actions (with more than one successor state), the maximum number of transitions occurring for a state-action pair, and the average number of transitions occurring for a state-action pair.
	
	For \realcs, we typically have slightly less than 2 actions per state, usually no more than 5 and never more than 13. Only a third of the actions are probabilistic. Typically, we have between one and two successors, seldom more than 4 and never more than 8.
	
	Since we wanted to explore graph structures that are not present in \realcs, for the generation of \many\ we allowed our algorithm to create models with more actions and transitions and higher branching. We typically have around 7 actions per state, but also up to 95. Three quarters of the actions are probabilistic.
	We have an average of slightly more than 2 transitions per action and a maximum of 10. The resulting models are more connected and thus often harder to solve for the algorithms.
	
	\item SmallestTransProb: The smallest transition probability occurring in the model.
	In \realcs, this goes from 0.5 (the largest possible non-trivial probability) to $10^{-3}$. 
	In \many, we always set this to $10^{-4}$, making the models generally hard for value iteration, so that non-trivial solution times occur.
	
	\item NumMECs, BiggestMEC\%, AvgMEC\%, NumSCCs, BiggestSCC\%, AvgSCC\%: These features are related to the important graph theoretic concepts of MEC and SCC. We give their number (Num), the size of the biggest occurring MEC/SCC in percentage of the state space (Biggest\%) as well as the average size of MECs/SCCs. Note that we only count MECs in the unknown part of the state space.
	
	Most of the models in \realcs\ have very few and very small MECs, with a few exceptions going up to around 2000 MECs. 
	Since we often only have one MEC, the biggest and average MEC size box plot look very similar. 
	In contrast, in \many\ we usually have around 2 non-trivial MECs, but they also are small. To analyze the impact of many or big MECs, we used the handcrafted model MulMec. However, we can also offer guidelines to create random models with certain numbers and sizes of MECs, analogous to the RandomSCC guideline.
	
	For SCCs, note that the scale of the plot for the number of SCCs is multiplied by one million. In \realcs\ we usually have around 100,000 SCCs, with an outlier having more than a million. The biggest SCC is often large, comprising a third of the model on average and going up to 50\% in three quarters of the cases. Since there are many transient states (which form a singleton SCC of size 1), the average SCC size is small. Since this does not give a lot of information, we also give the average size of non-singleton SCCs, see below.
	In \many, the variance is smaller. There are around 1,000 SCCs, and the biggest SCC makes up around 90\% of the model, making it almost completely strongly connected. This is because our random generation picks successor states randomly, and chances are good that this induces cycles throughout most of the model. This is why we offer the RandomSCC guideline, which allows to create models that less connected.

	\item MaxSCCDepth: This is important for topological algorithms, as it is the depth of the DAG forming the graph. Note that Example~\ref{ex:tvi} only needed a chain of 20 SCCs to show the problems of topological value iteration.
	In \realcs, the average depth is around 100 and it can go up to 600. In \many, this depth is low, because long chains of SCCs are unlikely to occur with the random generation. We analyze these chains using the handcrafted model from Appendix~\ref{app:exp:pt}. Alternatively, we could use the RandomSCC guideline.
	
	\item NumNonSingleton, SmallestSCCNonSing, AvgSCCNonSing\%: Since the many singleton (i.e.,\ trivial one state) SCCs make the features relating to all SCCs hard to interpret, these features analyze only the non-trivial SCCs. 
	In \realcs, in three quarters of the models there are very few non-trivial SCCs. 
	From the SmallestSCCNonSing (again with the axis multiplied with a million) and AvgSCCNonSing\% plot (which looks very similar to that for biggest overall SCC) we can see that in half the models even non-trivial SCCs are small. However, there are models where even the smallest SCC is big, indicating that there is one big SCC with chains of transient states around it.
	This is also the structure of the models from \many, where the singleton smallest and average size plot are the same as the overall biggest SCC plot.
\end{itemize}

\FloatBarrier
\section{Additional details on the experimental evaluation}\label{app:exp}

\subsection{Details on the optimizations}\label{app:exp:opts}

In this appendix, we describe the results of our evaluation of the different optimizations.
In principle, we can enhance every VI algorithm with any combination of them (with the exception that combining T and PT is the same as just having PT).

\begin{itemize}
	\item G: The Gau{\ss}-Seidel variant of VI does not perform the Bellman update on all states at once, but rather proceeds state by state. This allows to immediately use the new estimates of the states that were updated before and can thus speed up the computation.
	\item D: Deflating is a costly operation, since computing the SECs always requires a MEC decomposition. While this is possible in polynomial time~\cite{CY95}, it is slow in practice. 
	Thus, delaying deflating and only applying it only every $n$ steps can speed up the computation. However, for high $n$, the algorithm can also waste time waiting for the next deflation to happen. 
	Deflating every 100 steps was observed to be a good compromise in~\cite{KKKW18}.
	However, for OVI, we omitted this optimization, as there the whole point of the verification phase is to decrease the upper bound, and there is not point in waiting (for e.g.,\ the lower bound to converge).
	\item T: The topological variant of VI~\cite{TVI1} first computes a partitioning of the state space into SCCs.
	Instead of solving the whole SG at once, topological VI proceeds SCC by SCC, starting from those at the bottom of the topological ordering. This allows for using less memory at once and can speed up the computation.
	\item PT: The precise variant of topological VI which we introduced in Section~\ref{sec:pt}. For a description of how to read scatter plots, see Section~\ref{sec:exp:precise-details}.
\end{itemize}

For assessing the impact of the G, D and T optimizations on the algorithms' performance, we compare the optimized versions of BVI, OVI, and $\SI$ to their unoptimized ones, as shown in the scatter plots (Figures~\ref{fig:scatterG}, \ref{fig:scatterD} and \ref{fig:scatterT}).
PT is analyzed in Section~\ref{sec:exp:precise-details}.

\subsubsection{G optimization for BVI and OVI}:
Figure~\ref{fig:scatterG} indicates that using the Gau{\ss}-Seidel optimization may reduce both the time required to solve a model by 1-4 times in most cases.
The Gau{\ss}-Seidel optimization is slower in some cases because the values are computed sequentially to enable the use of the already computed results.
Also, the unoptimized version uses vector operations instead, which turn out to be faster sometimes.
Furthermore, it is possible that there are more iterations required to solve. 
This is because OVI and BVI may find different ECs to deflate depending on whether Gau{\ss}-Seidel is used or not.
In some cases, the unoptimized version is able to find more favourable sets of ECs and thus requires less iterations.
Note that the order in which the states are updated by Gau{\ss}-Seidel VI is random. However, changing the order of computation of the states to be a reverse topological enumeration of the states did not yield improvements in our experiments.  

\begin{figure}
        \centering
       	\includegraphics[width=0.5\textwidth]{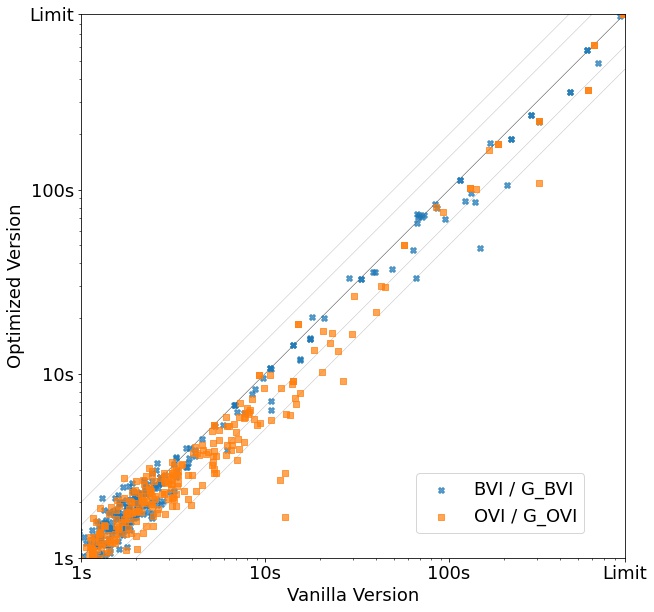}
        \caption{Scatter plot with logarithmic scale, comparing verification times of vanilla (unoptimized) BVI, OVI and with their Gau{\ss}-Seidel variant. Below the diagonal means with G is better, above means vanilla is better.}
        \label{fig:scatterG}
\end{figure}
\FloatBarrier

\subsubsection{D optimization for BVI}:

Figure~\ref{fig:scatterD} clearly indicates that although DBVI may sometimes solve models faster, for the majority of our models it could not compete with BVI.
\begin{figure}
        \centering
       	\includegraphics[width=0.5\textwidth]{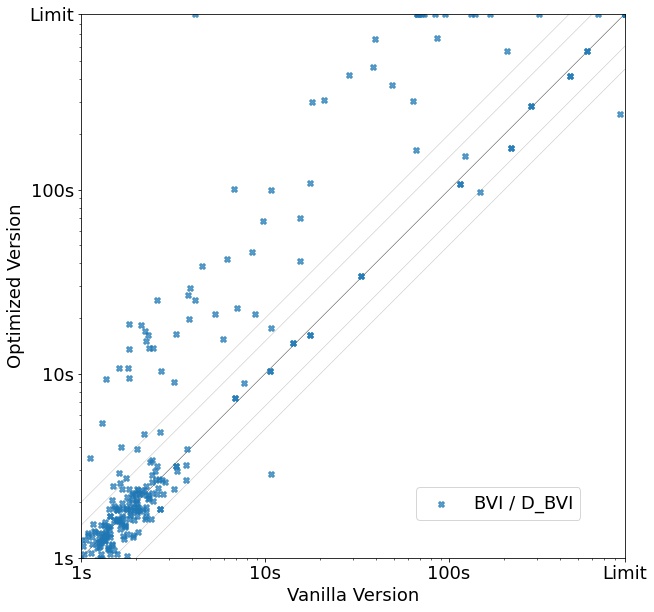}
        \caption{Scatter plot with logarithmic scale, comparing verification times of vanilla (unoptimized) BVI, OVI and $\LPSI$ with their deflating variant. Below the diagonal means with D is better, above means vanilla is better.}
        \label{fig:scatterD}
\end{figure}
\FloatBarrier

\subsubsection{T optimization for BVI, OVI and $\LPSI$}:

As the scatter plot in Figure~\ref{fig:scatterT} shows, on both types of models,
the topological addition to SI neither increases nor decreases its performance considerably.
However, most models have very few SCCs, so the topological optimization does not contribute a lot.
The data point where $\SI$ is significantly faster is on the real case study "dice", where every state is an SCC on its own.
Obviously, this is the best-case scenario for topological algorithms.

Adding the topological optimization to BVI often makes it significantly worse, mainly because of the precision issues addressed, when introducing PTBVI in Section~\ref{sec:pt}.
For OVI, the picture is slightly more mixed, but using the topological variant often helps.

\begin{figure}
        \centering
       	\includegraphics[width=0.5\textwidth]{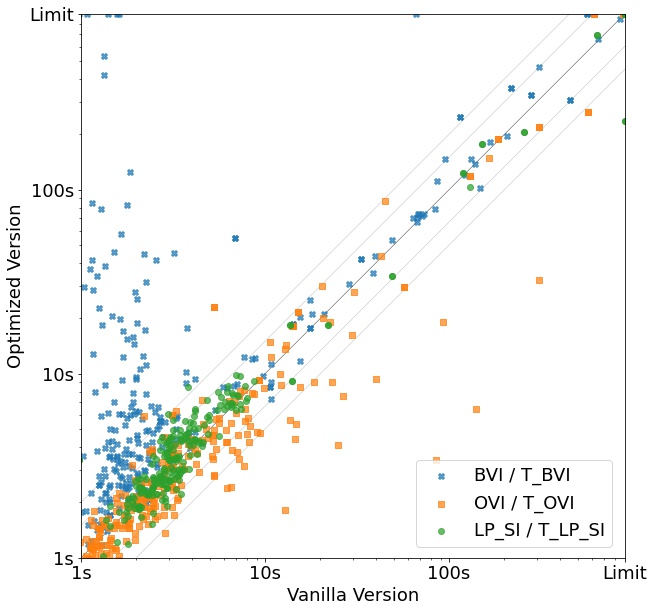}
        \caption{Scatter plot with logarithmic scale, comparing verification times of vanilla (unoptimized) BVI, OVI and $\SI$ with their topological variant. Below the diagonal means with T is better, above means vanilla is better.}
	\label{fig:scatterT}
\end{figure}
\FloatBarrier

\subsection{Additional scatter plot for BVI compared to WP and OVI}\label{app:exp:bvi}

Figure~\ref{fig:BVIvsWPvsOVI} shows the scatter plots comparing BVI with WP and OVI, complementing those in Figure~\ref{fig:OVIvsWPvsBVI}.
These two figures give all relations between the three algorithms.

\begin{figure}[h!]
	\centering
	\subfloat[
	Real and handcrafted models]{{\includegraphics[width=.45\textwidth]{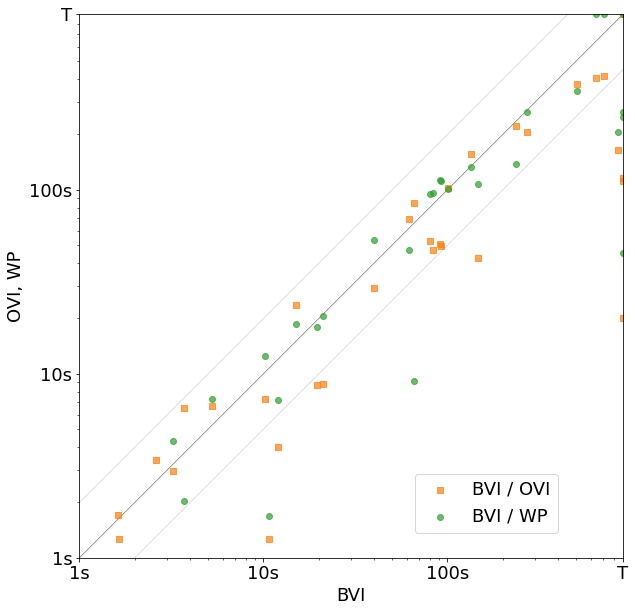} }}%
	\
	\subfloat[
	Random models]{{\includegraphics[width=.45\textwidth]{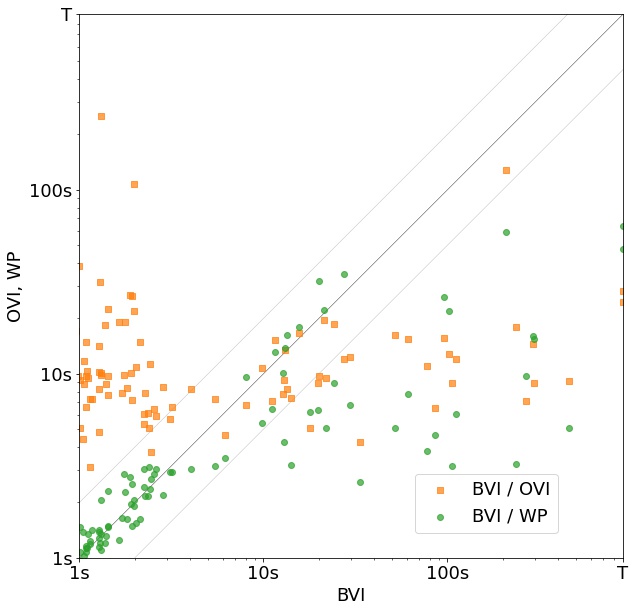} }}%
	\caption{BVI compared to WP and OVI.}%
	\label{fig:BVIvsWPvsOVI}%
\end{figure}
\FloatBarrier

\subsection{Additional details and plots for OVI}\label{app:exp:add}

\subsubsection{Handcrafted models for OVI.}
As an extreme example where the lower bound converges faster, consider a chain of Maximizer states, where every state has two actions: One that immediately yields a value of 0.5 and one that continues to the next state in the chain with low chance (0.01) and self-loops with a high chance (0.99). 
These kind of self loops always slow down VI algorithms, see also the haddad-monmege model~\cite{bvi}.
The last state in the chain has an action with value 0.49.
The lower bound of all states can immediately be set to 0.5, as all states have the first action that guarantees this value. However, to ensure that continuing in the chain is certainly the worse option, i.e.,\ to have a convergent upper bound, BVI has to wait for the information to propagate over the self-loops.
In contrast, OVI quickly knows the correct lower bound and then can verify it.
Concretely, BVI cannot solve the model within 2 minutes as soon as the chain has more than 200 states, while OVI can deal with more than 10,000 states, see Figure~\ref{fig:oviHandcrafted}.

As noted in Section~\ref{sec:exp:approx-details}, the dual situation where the lower bound converges slowly, is problematic for OVI. 
For example, if we remove the action with the value of 0.5 from the game above, we get a Markov chain where every state has a low chance
of progressing towards the target\footnote{Note that this is different from the completely adversarial example of\cite{bvi},
where every state does not only have self-loop, but even goes back to the initial state.}. On this type of chain, OVI is consistently around 4 times slower than BVI, as 
it is more time expensive to wait for the lower bound to converge than to converge from both sides.

\subsubsection{Number of verification phases}
Theorerically, in the worst case the number of iterations for the first verification phase is $\frac{1}{\varepsilon'}$ (which is 1,000,000 in our case). If the verification fails, $\varepsilon'$ is halved and number of iterations in the next verification phase doubles. Thus, it is possible that the verification phase is aborted, although the game is almost solved. OVI will iterate unnecessary extra steps until it reaches its new precision that verifies that OVI may indeed terminate.
For large models, it is more likely that a model is so complex that it requires multiple verification phases.
Thus, it is also more probable that the precision is halved at a point that will lead to unnecessary iterations.
In addition to that, iterations are more costly for larger models since the state space is large.

\begin{figure}[ht]
	\centering
	\includegraphics[width=1\textwidth]{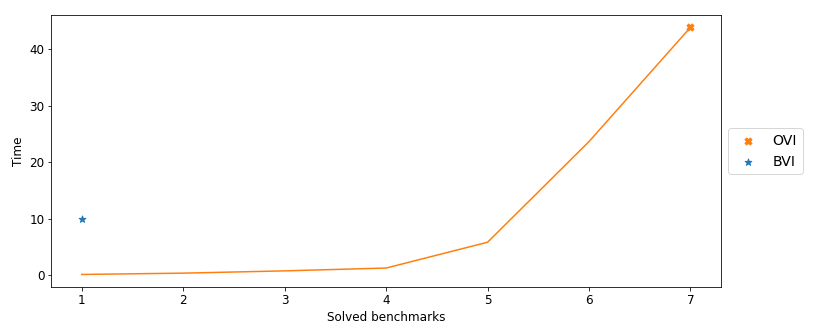}
	\caption{Measuring number of handcrafted OVI-model solved against aggregated runtime. The legend sorts the algorithms by their aggregated runtimes in descending order. The model is always the same, but with 100, 500, 900, 1000, 5000, 9000 and 10000 states. OVI is fast on this model, while BVI without optimizations cannot solve the model as soon as it has over 200 states. Topological optimizations handle the model very well, since every state is an SCC.}
	\label{fig:oviHandcrafted}
\end{figure}

\FloatBarrier

\subsubsection{Choice of parameters}

Two important parameters are the number of steps in a verification phase and the modification of the precision after a verification phase.

If a verification phase cannot verify the upper bound, we should abort it as soon as possible. However, it might be necessary to iterate for a long number of steps before the decrease of the upper bound propagates to all states, since it might be ``hidden'' behind some state switching its action.
Thus, we may also not choose it too small.

After a failed verification phase, we have to increase the lower bound enough that trying to guess the next upper bound has good chances of being correct. However, if the precision becomes too small, the VI-phase might have to run for a very long time before another verification phase is attempted.
On the other hand, keeping the precision too similar results in lots of aborted verification phases.

Finding a good trade-off, possibly even by dynamically changing these parameters during a run of the algorithm, is a task we leave for future work.

\subsection{Analysis of large models}\label{app:exp:pt}
To analyze how the algorithms scale on a large model with many SCCs, we handcrafted a model, called simple\_$n$\_$m$\_SCC.
It contains $n$ states and $m$ SCCs.
Every SCC forms a tree. The inner nodes have deterministic actions, leading to the next level. The leaves have probabilistic actions, leading to the root of the current or of the next tree, making the tree strongly connected.

Table~\ref{fig:tablePTBVI} shows the number of states and SCCs (visible in the name of the model) and the resulting verification times. 
We compared PTBVI with the fastest approximate algorithm, WP, and two variants of topological strategy iteration: 
one using linear programming for solving the opponent MDP ($\TLPSI$) and one using strategy iteration for solving the opponent MDP ($\TSISI$).

We see that PTBVI scales well on these models, while both variants of SI struggle time out on models with more than 2 million states.

\begin{table}[]
\centering
\begin{tabular}{l|rrrr}
	\toprule
	\textbf{Model name}       & \textbf{PTBVI}   & \textbf{$\TLPSI$}            & \textbf{$\TSISI$}           & \textbf{WP} \\ \midrule
simple\_50000\_1\_SCC                                     & 2.154                                  & \multicolumn{1}{r}{5.773}   & \multicolumn{1}{r}{7.079}  & 2.135                       \\ 
simple\_50000\_5\_SCC                                      & 4.564                                  & \multicolumn{1}{r}{4.999}   &                             & 10.548                      \\ 
simple\_100000\_1\_SCC                                    & 6.049                                  & \multicolumn{1}{r}{13.189}  & \multicolumn{1}{r}{14.694} & 4.33                        \\ 
simple\_100000\_5\_SCC                                    & 13.179                                 & \multicolumn{1}{r}{7.889}   & \multicolumn{1}{r}{12.927} & 25.34                       \\ 
simple\_500000\_1\_SCC                                    & 27.422                                 & \multicolumn{1}{r}{133.114} &                             & 22.127                      \\ 
simple\_500000\_5\_SCC                                  & 73.477                                 & \multicolumn{1}{r}{64.124}  &                             & 160.747                     \\ 
simple\_1000000\_1\_SCC                                 & 53.707                                 & \multicolumn{1}{r}{274.134} &                             & 53.319                      \\ 
simple\_1000000\_5\_SCC                                  & 156.327                                & \multicolumn{1}{r}{138.902} &                             & 401.318                     \\ 
simple\_2000000\_1\_SCC                                  & 103.149                                &                              &                             & 134.549                     \\ 
simple\_2000000\_5\_SCC                                  & 263.625                                & \multicolumn{1}{r}{366.236} &                             & 753.834                     \\ 
simple\_3000000\_1\_SCC                                & 165.242                                &                              &                             & 227.212                     \\ 
simple\_3000000\_5\_SCC                               & 389.154                                &                              &                             & 1244.764                    \\ 
simple\_4000000\_1\_SCC                                  & \multicolumn{1}{l}{}                  &                              &                             & 190.158                     \\ 
simple\_4000000\_5\_SCC                                 & 470.012                                &                              &                             & \multicolumn{1}{l}{}       \\ 
simple\_5000000\_1\_SCC                                 & \multicolumn{1}{l}{}                  &                              &                             & 301.517                     \\ 
simple\_5000000\_5\_SCC                                & 589.712                                &                              &                             & 1723.37                     \\ 
simple\_10000000\_1\_SCC                                & 581.999                                &                              &                             & 804.652                     \\ 
simple\_10000000\_5\_SCC                             & 1433.221                               &                              &                             & \multicolumn{1}{l}{}       \\ \bottomrule
\end{tabular}
	\caption{Performance comparison of PTBVI against $\TLPSI$, $\TSISI$, WP}
	\label{fig:tablePTBVI}
\end{table}
\end{appendix}

\end{document}